\documentclass[11pt,letterpaper]{article}

\usepackage{amsmath,amsthm,nicefrac}
\usepackage{amsfonts,amstext}
\usepackage{bm}
\usepackage{subfigure}
\usepackage{hhline}
\usepackage[table,dvipsnames]{xcolor}
\usepackage{tikz}
\usepackage{nicefrac}
\usepackage[margin=1in]{geometry}

\usepackage{amssymb}
\usepackage{mathrsfs}
\usepackage{xspace}
\usepackage{soul}
\usepackage{latexsym}
\usepackage{bbm}
\usepackage{enumitem}
\usepackage{aliascnt}

\usepackage[pdftex, plainpages = false, pdfpagelabels, 
bookmarks=false,
bookmarksopen = true,
bookmarksnumbered = true,
breaklinks = true,
linktocpage,
pagebackref,
colorlinks = true,  %
allcolors=blue!50!black,
hyperindex = true,
hyperfigures
]{hyperref} 
\usepackage{cleveref}
\usepackage{mathtools}

\crefname{table}{table}{tables}
\Crefname{table}{Table}{Tables}
\Crefname{figure}{Figure}{Figures}
\Crefname{algocf}{Algorithm}{Algorithms}
\crefname{algocfline}{line}{lines}
\crefname{theorem}{theorem}{theorems}
\Crefname{theorem}{Theorem}{Theorems}
\crefname{lemma}{lemma}{lemmas}
\Crefname{lemma}{Lemma}{Lemmas}
\crefname{claim}{claim}{claims}
\Crefname{invariant}{Invariant}{Invariants}
\Crefname{claim}{Claim}{Claims}
\crefname{fact}{fact}{facts}
\Crefname{fact}{Fact}{Facts}
\crefname{corollary}{corollary}{corollaries}
\Crefname{corollary}{Corollary}{Corollaries}
\crefname{assumption}{assumption}{assumptions}
\Crefname{assumption}{Assumption}{Assumptions}
\crefname{subclaim}{subclaim}{subclaims}
\Crefname{subclaim}{Subclaim}{Subclaims}
\crefname{defn}{definition}{definitions}
\Crefname{defn}{Definition}{Definitions}
\crefname{remark}{remark}{remarks}
\Crefname{remark}{Remark}{Remarks}
\crefname{appendix}{appendix}{appendices}
\Crefname{appendix}{Appendix}{Appendices}

\usepackage{setspace}

\allowdisplaybreaks

\newtheorem{theorem}{Theorem}[section]
\newaliascnt{lemma}{theorem}
\newtheorem{lemma}[lemma]{Lemma}
\aliascntresetthe{lemma}
\newaliascnt{claim}{theorem}

\aliascntresetthe{claim}
\newaliascnt{fact}{theorem}

\aliascntresetthe{fact}
\newaliascnt{corollary}{theorem}
\newtheorem{corollary}[corollary]{Corollary}
\aliascntresetthe{corollary}
\newaliascnt{assumption}{theorem}

\aliascntresetthe{assumption}

\newaliascnt{subclaim}{theorem}

\aliascntresetthe{subclaim}

\theoremstyle{definition}
\newaliascnt{defn}{theorem}
\newtheorem{defn}[defn]{Definition}
\aliascntresetthe{defn}

\theoremstyle{remark}
\newaliascnt{remark}{theorem}

\aliascntresetthe{remark}

\usepackage[font={small,it}]{caption}

\usepackage[ruled,vlined,noend,linesnumbered]{algorithm2e}

\SetCommentSty{mycommfont}

\newcommand{\cA}{\mathcal{A}}
\newcommand{\cB}{\mathcal{B}}
\newcommand{\cC}{\mathcal{C}}

\newcommand{\cM}{\mathcal{M}}

\newcommand{\cT}{\mathcal{T}}

\newcommand{\ceil}[1]{\left\lceil #1 \right\rceil}

\newcommand{\sse}{\subseteq}
\newcommand{\eps}{\varepsilon}

\newcommand{\nf}{\nicefrac}

\newcommand{\ignore}[1]{}

\newcommand{\initOneLiners}{%
    \setlength{\itemsep}{0pt}
    \setlength{\parsep }{0pt}
    \setlength{\topsep }{0pt}
}

\DeclareMathOperator{\poly}{poly}

\DeclareMathOperator{\cost}{cost}
\DeclareMathOperator{\Ball}{Ball}

\DeclareMathOperator{\OPT}{\mathsf{OPT}}
\DeclareMathOperator{\opt}{\mathsf{opt}}

\newcommand{\apx}{\ensuremath{\mathsf{apx}}\xspace}
\newcommand{\APX}{\ensuremath{\mathsf{APX}}\xspace}

\newcommand{\radius}{\operatorname{rad}}
\newcommand{\diam}{\operatorname{diam}}

\renewcommand{\emptyset}{\varnothing}

\newcommand{\MSD}{\textsf{MSD}\xspace}
\newcommand{\MSR}{\textsf{MSR}\xspace}

\newcommand{\recMSD}{\ensuremath{\mathtt{recMSD}}\xspace}
\newcommand{\recMSR}{\ensuremath{\mathtt{recMSR}}\xspace}
\newcommand{\detRecMSD}{\ensuremath{\mathtt{detRecMSD}}\xspace}
\newcommand{\DefensiveClustering}{\ensuremath{\mathtt{DefensiveClustering}}\xspace}

\title{FPT Approximation Schemes for Min-Sum Radii\\ and Min-Sum Diameters Clustering}

\author{
  Fabrizio Grandoni\thanks{IDSIA, USI-SUPSI. Email: \texttt{fabrizio.grandoni@supsi.ch}. Work supported in part by the SNSF grant 200021-236706.}
  \and
  Anupam Gupta\thanks{Computer Science Department, New York University. Email: \texttt{anupam.g@nyu.edu}. Work supported in part by NSF awards CCF-2422926 and CCF-2608359.}
  \and
  Jatin Yadav\thanks{IIT Delhi. Email: \texttt{jatin.yadav@cse.iitd.ac.in}. Work supported in part by a Google PhD Fellowship.}
}
\date{}
\begin{document}

\maketitle

\begin{abstract}
  \noindent In the classical Min-Sum Radii problem (\MSR) we are given
  a set $X$ of $n$ points in a metric space and a positive integer
  $k\in [n]$. Our goal is to partition $X$ into $k$ subsets (the
  \emph{clusters}) so as to minimize the sum of the radii of these
  clusters. The Min-Sum Diameters problem (\MSD) is defined
  analogously, where instead of the radii of the clusters we consider
  their diameters.

  For both problems we present FPT approximation schemes for the
  natural parameter $k$. Specifically, given $\eps>0$, we show how to
  compute $(1+\eps)$-approximations for both \MSD and \MSR in time
  $(1/\eps)^kn^{O(1)}$ and $(1/\eps)^{O(k/\eps \log 1/\eps)}n^{\poly(1/\eps)}$
  respectively. The previous best FPT
  approximation algorithms for these problems have approximation
  factors $4+\eps$ and $2+\eps$, respectively,
  and finding an FPT approximation scheme for both these problems had been outstanding open problems.
\end{abstract}

\thispagestyle{empty}

\setcounter{page}{1}

\section{Introduction}
\label{sec:introduction}

We consider two classical clustering objectives closely related to
each other: \emph{min-sum radii} and \emph{min-sum diameters}, which
seek to partition the points of a given metric space into $k$
clusters, such that the sum of their radii or diameters is
minimized. Both these problems have long been known to be NP-hard, but
their approximability has been a source of much interest, as we detail
below. In this work, we study these questions from the perspective of
\emph{parameterized complexity}, and ask: \emph{what can we do in time
  $f(k) \poly(n)$, where $k$ is the number of clusters and $n$ the
  size of the metric space?}

Formally, both the Min-Sum Radii and Min-Sum Diameters problems are
defined on a (symmetric) metric space $\cM = (X,d)$ with $X$ being a
collection of $n$ \emph{points}. Given an integer $k\in [n]$ denoting
the target number of clusters, the goal in both problems is to
partition the point set $X$ into $k$ sets (also called
\emph{clusters}) $S_1, S_2, \ldots, S_k$. The difference between the
two problems is in their objective functions:
\begin{itemize}
\item The diameter of a cluster $S$ is the maximum distance between
  any pair of points in $S$---i.e.,
  \[ \diam(S) = \max_{p,q\in S} d(p,q).\] In the Min-Sum Diameters (\MSD)
  problem, the goal is to minimize $\sum_{i \in [k]} \diam(S_i)$.
\item The \emph{center} (or the $1$-center) of a cluster $S$ is a
  point $p \in X$ minimizing $\max_{q \in S} d(p,q)$; the latter
  quantity is called the \emph{radius} of the
  cluster.%
    In other words,
  \[ \radius(S) := \min_{p \in X} \max_{q \in S} d(p,q). \] In the
  Min-Sum Radii (\MSR) problem, the goal is to minimize
  $\sum_{i \in [k]} \radius(S_i)$.
\end{itemize}

Both these problems are 
NP-hard~\cite{10.1007/3-540-44985-X_22,10.1007/978-3-540-69903-3_26},
and hence have been studied using the lens of approximation
algorithms. Over a sequence of works giving polynomial-time approximation algorithms~\cite{10.1007/3-540-44985-X_22,CHARIKAR2004417,friggstad_et_al:LIPIcs.ESA.2022.56,doi:10.1137/1.9781611977912.69}, 
the approximation ratio for \MSR has been improved  to $3 + \eps$;
since the two objectives are separated by
a factor of at most $2$, this implies a $(6+\eps)$-approximation for
\MSD as well, which is currently the best approximation factor
known. The \MSD problem does not admit a polynomial time algorithm
with an approximation factor of better than $2$, unless
$P=NP$~\cite{10.1007/3-540-44985-X_22}; however, \MSR admits a
quasi-polynomial time approximation scheme
(QPTAS)~\cite{10.1007/978-3-540-69903-3_26}, and hence a polynomial
time approximation scheme (PTAS) is likely to exist.

\MSR and \MSD are also well-studied in
terms of parameterized approximation algorithms, where we allow the
running time to depend super-polynomially in some parameter. For the
natural parameter $k$, the best known results are a $2+\eps$ (resp.,
$4+\eps$) approximation for \MSR (resp., \MSD) in FPT time
$f(k,\eps)\cdot n^{g(\eps)}$ for some computable functions $f, g$
\cite{bandyapadhyay_et_al:LIPIcs.SoCG.2023.12,10.1609/aaai.v38i18.30053}.
In this work, our focus is on \emph{FPT approximation schemes}---i.e.,
$(1+\eps)$-approximation algorithms running in time
$f(k, \eps) \cdot n^{\poly(1/\eps)}$ for some computable function
$f$. %
The following natural question
remains unresolved: \emph{do these two problems admit FPT
  Approximation Schemes?}

\subsection{Our Results}
\label{sec:our-results}

The main results of our work answer the above question in the
affirmative. For the \MSR problem, which is known to be $W[1]$-hard~\cite{10.1609/aaai.v39i15.33699}, we substantially improve upon the mentioned $(2 + \eps)$-approximation in FPT time
\cite{bandyapadhyay_et_al:LIPIcs.SoCG.2023.12,10.1609/aaai.v38i18.30053}.

\begin{theorem}[Min-Sum Radii]
  \label{thr:mainMSR}
  For any $\eps>0$, one can compute a $(1+\eps)$-approximation for
  \MSR in deterministic time
  $(1/\eps)^{O(\nicefrac{k}{\eps} \ln \nicefrac{1}{\eps})}n^{O(\nicefrac{1}{\eps^3}\ln \nicefrac{1}{\eps})}$.
\end{theorem}

We then turn our attention to the \MSD problem. It is remarkable that
essentially all previous approximation algorithms for this problem
were obtained by losing a factor of $2$ over the results for \MSR. In
other words, the best approximation factors were $6+\eps$ and $4+\eps$
in polynomial and FPT time respectively. We manage to break past this
issue, and give a $(1+\eps)$-approximation in FPT time:

\begin{theorem}[Min-Sum Diameters]
  \label{thr:mainMSD}
  For any $\eps>0$, one can compute a $1+\eps$ approximation for \MSD
  in randomized time $(O(1/\eps))^k n^{O(1)}$ and in deterministic
  time $(O(k/\eps))^k n^{O(1)}$.
\end{theorem}

Note that such a result is not possible in polynomial time: an
approximation of better than a factor of $2$ would imply that
$P=NP$~\cite{10.1007/3-540-44985-X_22}.
Our approach also gives an exact algorithm
with runtime $n^{k+O(1)}$ (see \Cref{sec:exact_msd}), thereby
improving the previous best runtime of $n^{5k + O(1)}$ given 
by~\cite{10.1609/aaai.v39i15.33699}. 

The same algorithm can also be utilized to give an FPT approximation scheme for the $\MSD$ problem (and thus a FPT $(2 + \eps)$-approximation algorithm for the \MSR problem) with \emph{mergeable constraints}. Here, we are given an upper bound $k$ on the number of clusters and each cluster is required to satisfy some constraint. The constraint is such that, if two disjoint subsets of points satisfy the constraint, then so does their union. 
Carta et al.~\cite{carta_et_al:LIPIcs.ISAAC.2024.16} show that fairness constraints like \emph{exact fairness} and \emph{ratio balance}, as well as non-fairness constraints such as lower bounds on cluster sizes, are mergeable.
The previous best known result for mergeable constraints is an FPT $(4+\eps)$-approximation for \MSR (and an implied FPT $(8+\eps)$-approximation for \MSD)~\cite{bandyapadhyay_et_al:LIPIcs.APPROX/RANDOM.2025.23}.

\begin{theorem}[Mergeable clustering]  \label{thr:mergeable_clustering}
  For any $\eps>0$, one can compute a $(1+\eps)$-approximation for \MSD (and thus a $(2 + \eps)$-approximation for \MSR) under any mergeable constraint (that is verifiable in polynomial time)
  in randomized time $(O(1/\eps))^k n^{O(1)}$ and in deterministic
  time $(O(k/\eps))^k n^{O(1)}$.
\end{theorem}

\subsection{Our Techniques}
\label{sec:our-techniques}

\subsubsection{The Algorithm for Min-Sum Diameters}

Our FPT approximation scheme for \MSD is surprisingly clean, once
viewed from the right perspective. To explain the guiding
  intuition, let us present a slight variant of our algorithm.  We
claim that the following algorithm has probability at least
$(\eps/k)^{O(k)}$ of returning a $1+O(\eps)$ approximate solution;
therefore, an FPT randomized approximation scheme can be obtained by
repeating it an FPT number of times. Our
recursive algorithm gets in input a subset of points $S\subseteq X$
and a parameter $k'\in [k]$, and returns a clustering of $S$ with at
most $k'$ clusters. With probability $\nf12$, the algorithm returns
$\{S\}$ itself as a clustering (this happens also when
$k'=1$). Otherwise, let $x,y\in S$ be a diameter witness pair, i.e.,
$d(x,y)=\diam(S)$. The algorithm chooses $R\in [0,\diam(S)]$ uniformly
at random and $k_1\in [k'-1]$ uniformly at random. Then it returns the
union of the two solutions obtained recursively for the subproblems
$(S_1,k_1)$ and $(S\setminus S_1,k'-k_1)$, where
$S_1=S\cap \Ball(x,R)$ and $\Ball(x,R)$ is a ball of radius $R$
centered at $x$ (containing all the points at distance at most $R$
from $x$). The desired clustering is obtained from the root call
$(X,k)$.

To see the main idea of the analysis, consider the special case of the
root call. Let $X^\star_1,\ldots,X^\star_k$ be the optimal clustering, so that
the optimal cost is $\opt=\sum_{j=1}^{k}\diam(X^\star_j)$. Consider the
diameter witness pair $x,y\in X$ chosen for $X$. Let us project (see
\Cref{fig:msd_image}) the points of $X$ along the interval
$I=[0,\diam(X)]$ such that each $p\in X$ is placed in position
$d(x,p)$ along $I$. In particular, $x$ and $y$ are projected to the
endpoints of $I$. Each optimal cluster $X^\star_j$ induces an interval
$I^\star_j$ of width at most $\diam(X^\star_j)$, where these intervals might
overlap. Crucially, the portion of $I$ spanned by such intervals has
length at most $\opt$. Then, if $\diam(X)\leq \frac{1}{1-\eps}\opt$, with probability $\nf12$ the
above algorithm returns $\{X\}$, which is a feasible $1+O(\eps)$
approximation.

Otherwise, the region of $I$ \emph{not} spanned by the intervals
$I^\star_j$ has length at least $\eps \cdot \diam(X)$. In this case, with
probability $\nf12$ the algorithm chooses to pick a random pair
$(R,k_1)$. Now, with probability at least $\eps$, $R$ falls in the
region of $I$ not spanned by the intervals $I^\star_j$. Assuming that this
happens, each optimal cluster $X^\star_j$ is either entirely at distance
at most $R$ from $x$ (hence it is fully contained in $S_1$ for this
recursive call), or it is entirely at distance more than $R$ from $x$
(hence it is entirely contained in $X\setminus S_1$). Furthermore,
with probability $\frac{1}{k-1}$ the chosen value $k_1$ for this
recursive call is exactly the number of optimal clusters fully
contained in $S_1$. Therefore, the two subproblems $(S_1,k_1)$ and
$(X\setminus S_1,k-k_1)$ partition the problem in two subproblems that
can be solved independently.

The actual algorithm in \S\ref{sec:algorithm-min-sum} is an
  optimized and slightly crisper version of the above description. Moreover, we can derandomize it at a small additional cost: the
details are given in \Cref{sec:derandomization}.

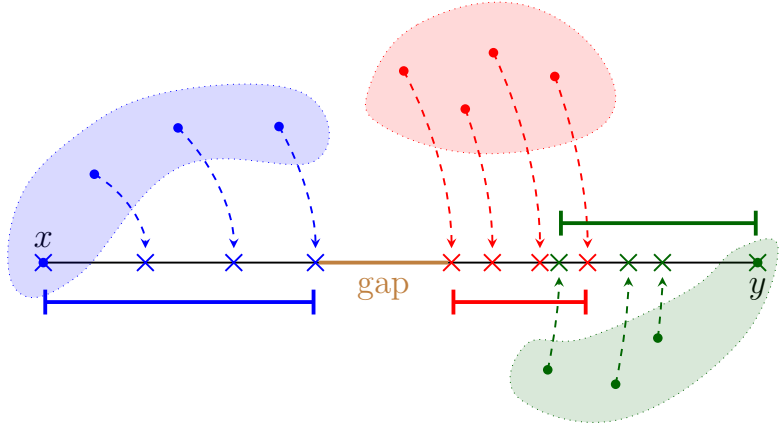
\begin{figure}[t]
  \centering
  \scalebox{0.9}{
\begin{tikzpicture}[
    point/.style={circle, fill, minimum size=4pt, inner sep=0pt},
    arc style/.style={dashed, ->, >=stealth, thick, shorten >= 6pt},
    interval/.style={ultra thick, |-|},
    cross/.pic={
        \draw[thick, pic actions] (-3.5pt,-3.5pt) -- (3.5pt,3.5pt) (-3.5pt,3.5pt) -- (3.5pt,-3.5pt);
    }
]

\draw[thick] (0, 0) -- (10.5, 0);

\node[above] at (0,0.1) {\Large $x$};
\pic[blue] at (0,0) {cross};

\pic[green!40!black] at (10.5,0) {cross};

\foreach \a/\r in {60/1.5, 45/2.8, 30/4.0} {
    \node[point, blue] (P) at (\a:\r) {};
    \draw[blue, arc style] (P) arc (\a:0:\r);
    \pic[blue] at (\r,0) {cross};
}
\draw[blue, interval] (0, -0.58) -- (4.0, -0.58);

\foreach \a/\r in {28/6.0, 20/6.6, 25/7.3, 20/8.0} {
    \node[point, red] (P) at (\a:\r) {};
    \draw[red, arc style] (P) arc (\a:0:\r);
    \pic[red] at (\r,0) {cross};
}
\draw[red, interval] (6.0, -0.58) -- (8.0, -0.58);

\foreach \a/\r in {-12/8.6, -7/9.1, -12/7.58} {
    \node[point, green!40!black] (P) at (\a:\r) {};
    \draw[green!40!black, arc style] (P) arc (\a:0:\r);
    \pic[green!40!black] at (\r,0) {cross};
}
\draw[green!40!black, interval] (7.58, 0.58) -- (10.5, 0.58);

\node[point, blue] at (0,0) {};
\node[point, green!40!black] at (10.5,0) {};
\node[below] at (10.5,-0.1) {\Large $y$};

\draw[blue, fill=blue, fill opacity=0.15, dotted] plot [smooth cycle, tension=0.7] coordinates {
    (-0.4, -0.4) (-0.4, 0.8) (0.4, 1.8) (1.8, 2.5) (3.8, 2.5) (4.0, 1.5) (2.0, 1.4) (0.5, -0.2)
};

\draw[red, fill=red, fill opacity=0.15, dotted] plot [smooth cycle, tension=0.7] coordinates {
    (4.8, 2.6) (5.5, 3.5) (7.0, 3.8) (8.2, 3.0) (8.2, 2.0) (6.5, 1.6) (5.0, 2.0)
};

\draw[green!40!black, fill=green!40!black, fill opacity=0.15, dotted] plot [smooth cycle, tension=0.7] coordinates {
    (10.8, 0.2) (10.2, 0.2) (9.5, -0.4) (8.0, -1.1) (7.1, -1.2) (7.0, -2.0) (8.5, -2.3) (10.2, -1.2)
};

\draw[ultra thick, brown] (4.05, 0) -- (5.95, 0);
\node[below, brown] at (5.0, -0.1) {\Large gap};

\end{tikzpicture}
}
  \caption{Projecting each point on the diameter $(x, y)$. There are three clusters, blue, red and green. Each interval represents the range of distances from $x$ for the points in the corresponding cluster. The algorithm chooses a random threshold $R$ in $[0, \diam(S)]$ and splits the problem into left and right of the threshold. The created split is good if the threshold lies in the gap, that is, it does not lie on any of the three intervals.}
  \label{fig:msd_image}
\end{figure}

\subsubsection{The Algorithm for Min-Sum Radii} Our FPT
approximation scheme for \MSR is substantially more involved. It is
instructive to first see how to achieve a $2(1+\eps)$ approximation in
FPT time (as done by
\cite{bandyapadhyay_et_al:LIPIcs.SoCG.2023.12,10.1609/aaai.v38i18.30053}). Consider an optimal solution, which can be equivalently described as a set of balls $B_1^\star, B_2^\star, \ldots, B_k^\star$. We compute a set of balls that together cover all the points. Until there exists a yet-uncovered point $c_j$, we construct a new ball
$B_j:= \Ball(c_j,2r^\star_{\kappa(j)})$, where $r^\star_{\kappa(j)}$ is the
radius of the optimal cluster $B^\star_{\kappa(j)}$ containing
$c_j$. Notice that $B_j$ must entirely contain
$B^\star_{\kappa(j)}$. Repeating this process ensures that we create at
most $k$ balls covering all the points, whose total radius is at most
$2$ times the optimal cost $\opt$. The running time is affected by the
guessing of the values $r^\star_{\kappa(j)}$: this can be made FPT by a
simple preprocessing step (losing another factor $1+\eps$ in the
approximation).

Our improved FPT algorithm is based on the following  three main ideas:
\begin{enumerate}[label=(\alph*)]\itemsep0pt
\item \label{idea:1} If we can define a ball $D := \Ball(c,r)$ such that the total
  radius of the optimal clusters completely contained within $D$ is at
  least $r/(1+\eps)$, we can output the ball $D$ as part of the final solution, and charge the
  radius $r$ to the optimal clusters completely handled by it.

\item \label{idea:2} If there are a constant number $M$ of optimal balls whose total
  radius is an $\Omega(\eps)$ fraction of the optimal cost, we could guess
  the centers and radii of these $M$ balls in polynomial time,
  and then recurse on the remaining points. Crucially, the optimal
  cost of this residual instance is now smaller by a constant
  factor. This can only happen for a constant number of recursion levels before the
  residual cost is only an $\eps$ factor of $\opt$---whereupon we
  can just use any $O(1)$-approximation algorithm.

\item \label{idea:3} Suppose that we find a ball
  $E := \Ball(c,r)$ and a collection ${\cal B}$ of \emph{boundary} balls around $E$ (that intersect $E$ without being contained in it) such that the total radius of $\cal B$ is tiny, say $O(\eps^2 r)$. Then we can use a $O(1)$-approximation to cover the points inside ${\cal B}$, as
  long as inside $E$ we are able to compute balls of total radius $\Omega(\eps r)$ that are \emph{close} to optimal.
\end{enumerate}

Let us describe in more detail how to combine the above ideas. Our
starting step is the construction of an alternative set of
\emph{defensive} balls
$D_1, D_2, \ldots, D_{\ell}$, with each ball of the form
$D_j = \Ball(c_j, r_j)$, which together cover all the points $X$ (see
\Cref{fig:defensive_ball_image}). The construction is as follows. Initially all the
points are uncovered. Like in the $2+\eps$ approximation that we
described above and with a similar notation, we pick an arbitrary
uncovered point $c_j$ and initially define the next ball as
$D_j:=\Ball(c_j,2r^\star_{ \kappa(j)})$ (later we call $B^\star_{\kappa(j)}$
the \emph{core} of $D_j$). If $D_j$ includes the center of some other
optimal cluster $B^\star_h$ still containing uncovered points, we increase
the radius of $D_j$ by the radius $r^\star_h$ of the largest such
cluster. Notice that after this expansion, $B^\star_j$ is fully covered by
$D_j$. We continue similarly until no further $B^\star_h$ of the above
type is defined, therefore fixing the final radius $r_j$ of
$D_j$. Notice that the expansions of $D_j$ are very \emph{cheap} in
the sense that the corresponding \emph{cost} can be charged to the
radius of some (distinct) optimal clusters that are covered by $D_j$.
\begin{figure}[t]
  \centering
  \scalebox{0.7}{
  \begin{tikzpicture}
    \coordinate (O) at (0,0);

    \draw[solid, thick] (O) circle (4.5cm); %
    \draw[dashed, thin] (O) circle (2.8cm); %
    \draw[dashed, thin] (O) circle (1.5cm); %

    \node[text=black] at (0,-1.5) [below] {$\Ball(c_j, 2r_1^\star)$};
    \node[text=black] at (-0.3,-2.8) [below] {$\Ball(c_j,2r_1^\star+r_2^\star)$};
    \node[text=black, font=\large] at (0,-4.5) [below] {$D_j = \Ball(c_j, 2r_1^\star + r_2^\star + r_3^\star)$};

    \draw[->, thick] (-20:1.5cm) -- (-20:2.8cm) node[pos=0.5, above] {$r_2^\star$};

    \draw[->, thick] (-65:2.8cm) -- (-65:4.5cm) node[pos=0.5, right] {$r_3^\star$};

    \fill[black] (O) circle (1.5pt);
    \node at (O) [below] {$c_j$};
    \coordinate (B1_coord) at (-0.3, -0.5);
    \draw[orange, thick, fill=gray, fill opacity=0.15] (B1_coord) circle (0.78cm);
    \node[text=orange, font=\large] at (-0.8,-1) [above right, xshift=1mm] {$B_1^\star$};

    \coordinate (B2_coord) at (0.6, 1.1);
    \draw[orange, thick] (B2_coord) circle (1cm);
    \node at (B2_coord) [below] {$c_2^\star$};
    \fill[black] (B2_coord) circle (1.5pt);
    \node[text=orange, font=\large] at (1,1) [above, yshift=1mm] {$B_2^\star$};

    \coordinate (B3_coord) at (-2.3, -0.6);
    \draw[orange, thick] (B3_coord) circle (1.3cm); 
    \fill[black] (B3_coord) circle (1.5pt);
    \node at (B3_coord) [below] {$c_3^\star$};
    \node[text=orange, font=\large] at (-1.8,-0.5) [above left] {$B_3^\star$};

    \coordinate (B4_coord) at (1.9, -1.8);
    \draw[orange, thick] (B4_coord) circle (0.65cm); 
    \node[text=orange, font=\large] at (2.45,-1.9) [right] {$B_4^\star$};
    \fill[black] (B4_coord) circle (1.5 pt);
    \node at (B4_coord) [above] {$c_4^\star$};

    \coordinate (B5_coord) at (-2.4, 2.8);
    \draw[orange, thick] (B5_coord) circle (0.6cm); 
    \node[text=orange, font=\large] at (-1.8,3) [right] {$B_5^\star$};
    \fill[black] (B5_coord) circle (1.5 pt);
    \node at (B5_coord) [above] {$c_5^\star$};

    \draw[orange, thick] (-4.5, 4.0) circle (0.8cm);
    \fill[black] (-4.5, 4.0) circle (1.5pt);
    \draw[orange, thick] (4.5, -4.0) circle (1.2cm);
    \fill[black] (4.5, -4.0) circle (1.5pt);

    \draw[orange, thick] (-2.7, 3.8) circle (0.25cm);
    \fill[black] (-2.7, 3.8) circle (1.5pt);
    \draw[orange, thick] (4.5, 1.5) circle (1.0cm);
    \fill[black] (4.5, 1.5) circle (1.5pt);

  \end{tikzpicture}
}

\caption{Creation of a defensive ball $D_j$. The originally
  \emph{undefended} optimal balls are given in orange. We start with
  an uncovered point $c_j$. Since it is contained in the optimal ball
  $B_1^\star$, we start with $D_j \gets \Ball(c_j, 2r_1^\star)$, and
  mark $B_1^\star$ as defended. Since $B_2^\star$'s center $c_2^\star$ is inside
  $D_j$, we expand $D_j$ by $r_2^\star$ to get
  $D_j = \Ball(c_j, 2r_1^\star + r_2^\star)$ and mark $B_2^\star$ as
  defended by $D_j$. Now $B_3^\star, B_4^\star$ have their centers
  inside $D_j$, so we further expand $D_j$ by
  $\max(r_3^\star, r_4^\star)=r^*_3$ to make
  $D_j = \Ball(c_j, 2r_1^\star + r_2^\star + r_3^\star)$. Optimal
  balls $B_3^\star, B_4^\star, B_5^\star$ are all marked as defended
  by $D_j$. Since $D_j$ does not contain the center of any other
  undefended optimal ball, we stop. The core for $D_j$ is $B_1^\star$
  (the shaded optimal ball), and $D_j$ defends the optimal balls
  $B_1^\star, B_2^\star, B_3^\star, B_4^\star, B_5^\star$ and so
  $I(j) = \{1, 2, 3, 4, 5\}$.}
  \label{fig:defensive_ball_image}
\end{figure}
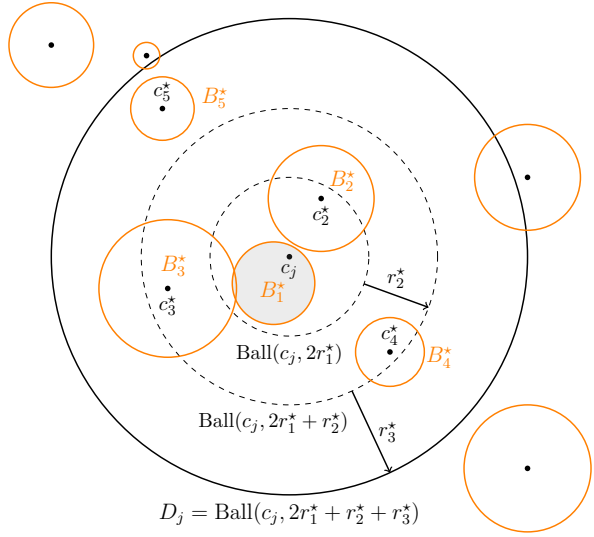
The (\emph{defensive}) clustering induced by the defensive balls has
already some useful properties:
\begin{enumerate}[nosep]
\item For each optimal cluster $B^\star_i$, there exists at least one
  defensive ball $D_j$ fully containing it. We say that $B^\star_i$ is
  assigned to the first $D_j$ that fully covers $B^\star_i$, and let
  $I(j)$ be the indexes of the optimal clusters assigned to $D_j$
  (including necessarily $\kappa(j)$).%
\item Observe that $r_j\leq r^\star_{\kappa(j)}+\sum_{i\in I(j)}r^\star_{i}$.  If
  we are lucky and $r_j\leq (1+O(\eps))\sum_{i\in I(j)}r^\star_{i}$, then
  we can use Idea~\ref{idea:1} to directly include $D_j$ in our
  solution; we say that such a $D_j$ is \emph{good}.
\end{enumerate}
However, the defensive balls may not satisfy the condition
$r_j\leq (1+O(\eps))\sum_{i\in I(j)}r^\star_{i}$. In that case, we call
such a defensive ball \emph{bad}. As an extreme case, it might happen
that $r_j=2r^\star_{\kappa(j)}$ for all the defensive balls $D_j$---we
need further ideas for this case, if we are to get a better-than-$2$
approximation.

The second part of our algorithm is designed to address the problem
due to the presence of bad defensive balls. In this second part, we
pick some (bad) defensive ball $D$, and \emph{expand} it to
incorporate other defensive balls, according to the following
construction (that has some similarities with the construction of $D$
itself, but with some critical differences). Let $E=\Ball(c,r')$ be
the current expanded ball (initially $E=D$), and consider the
\emph{boundary balls} ${\cal B}$ of $E$, i.e., the defensive balls
that intersect with $E$ without being fully covered by $E$. We
iteratively increase the radius of $E$ as long as one of the following
conditions holds:
\begin{itemize}
\item (\emph{cheap expansion}) We increase the radius of $E$ by the minimal amount $S$ such that 
  $E$ fully contain ${\cal B}$, under the assumption that the total radius of the latter balls is at least $\rho S$ for a carefully chosen constant $\rho\geq 2$. Notice that this expansion is \emph{cheap} since by paying an extra $S$ we incorporate defensive balls of total radius at least $2S$, and therefore optimal clusters of total radius at least $S$. 
  \item (\emph{large expansion}) If there exists a boundary ball $B\in {\cal B}$ of sufficiently large radius $r$ w.r.t. to the radius $r'$ of $E$, namely $r\geq \Omega(\eps^2) r'$, we increase the radius of $E$ by $2r$ for the largest such $r$. Notice that after this expansion $E$ fully includes $B$.  
\end{itemize}
We show that, for a constant $\rho$ large enough, a constant fraction of the final radius $r'$  of $E$ is due to large expansions (where we consider the initial $D$ as a large expansion). Even more, a constant fraction of $r'$ is only due to the radius of a \emph{constant} number $M$ of defensive balls that \emph{caused} such large expansions. Since the latter $M$ defensive balls are bad, this means that their respective cores have total radius $\Omega(\eps) r'$. This way we can apply Idea \ref{idea:2}: we can guess the mentioned cores and capture a constant fraction of the profit of all the optimal clusters fully contained in $E$.  

Furthermore, we show that the total radius of the final boundary balls ${\cal B}$ is at most $O(\eps^2 )r'$. Hence we can apply Idea \ref{idea:3}: we include such balls in the solution under construction and charge their cost to the cores that we guessed before. After removing the points covered by ${\cal B}$, the algorithm then just continues recursively on the remaining points $X'$ inside $E$ and in their complement $X''$ (modulo guessing the number of optimal clusters that still need to be fully covered on the two sides).

We remark that, by the first property of the defensive balls, after removing the boundary balls, the remaining part of each optimal cluster is fully contained in $X'$ or in $X''$. In particular, it does not happen that some optimal cluster has to be covered in part in the subproblem associated with $X'$ and in part in the one associated with $X''$ (which would cause at least a factor $2$ in the approximation ratio). This would not be true if we used, instead of the defensive balls, the balls as in the $2+\eps$ approximation that we described before.  

The final algorithm is actually slightly different for subtle
technical reasons, in particular, due to certain defensive balls covering the centers of some optimal balls that they do not defend; however, the above description captures the essential
ideas.

\subsection{Related Work}
\label{sec:related-work}

\textbf{Polynomial-time Approximate Algorithms.}  Both exact and
approximation algorithms for \MSR and \MSD problems have been
extensively studied. (Any $\alpha$-approximation for \MSR immediately
implies a $2\alpha$-approximation for \MSD, so we do not explicitly
mention all the implied results for \MSD henceforth.) For
approximation algorithms running in polynomial time, a sequence of
improvements~\cite{10.1007/3-540-44985-X_22,CHARIKAR2004417,friggstad_et_al:LIPIcs.ESA.2022.56}
culminated in a $(3+\eps)$-approximation for \MSR by Buchem et
al.~\cite{doi:10.1137/1.9781611977912.69}; they also give a
$(3+\eps)$-approximation for \MSR \emph{with outliers}, and a
$(3.5 + \eps)$-approximation for \MSR \emph{with outliers and
  generalized lower bounds}, where there is a lower bound for each
point specifying the number of points that must be assigned to it, if
it is made a cluster center.

\medskip\noindent\textbf{Exact Algorithms.} A remarkable result of
Gibson et al.~\cite{10.1007/978-3-540-69903-3_26} showed that the \MSR
problem can be solved exactly in $n^{O(\log n \log \Delta)}$ time,
where $\Delta$ is the aspect ratio of the metric space; they used this
algorithm to also give a QPTAS. In another
work~\cite{doi:10.1137/100798144}, they showed that \MSR can be solved
in polynomial time for Euclidean spaces of fixed dimension. The
runtime of the exact algorithm for \MSR in general metric spaces was
improved to $n^{O(\log n \log\log n + \log \Delta)}$ by Banerjee et
al.~\cite{10756173}. While there is a trivial polynomial time exact
algorithm for \MSR when $k$ is a constant (one can choose the centers
and the radii of all the clusters in $n^{O(k)}$ time), an exact
polynomial time algorithm for \MSD for a constant $k$ was first given
by Behsaz and Salavatipour~\cite{behsaz-savalatipour}, with a runtime
of $n^{O(k^2)}$. This was recently improved to $n^{5k +
  O(1)}$%
by~\cite{10.1609/aaai.v39i15.33699}, which we further improve to 
$n^{k+O(1)}$ in~\Cref{sec:exact_msd}.

\medskip\noindent\textbf{Parameterized (Approximate) Algorithms.}
In the setting of parameterized algorithms, the \MSR problem parametrized by the number of clusters admits a
$(2+\eps)$-approximation~\cite{bandyapadhyay_et_al:LIPIcs.SoCG.2023.12,10.1609/aaai.v38i18.30053};
to the best of our knowledge, this result (and the implied
$4+\eps$-approximation for \MSD) are the best previous results for
general metric spaces. For metric spaces with bounded doubling
dimension, a $(1+\eps)$-approximation in FPT time is
known~\cite{10.1609/aaai.v39i15.33699}. 
For capacitated \MSR, the works
of~\cite{inamdar_et_al:LIPIcs.ESA.2020.62,bandyapadhyay_et_al:LIPIcs.SoCG.2023.12,jaiswal_et_al:LIPIcs.ITCS.2024.65,filtser2025fptapproximationscapacitatedsum}
culminated in a factor of $(3+\eps)$ for uniform
capacities~\cite{jaiswal_et_al:LIPIcs.ITCS.2024.65} and $5.83$ for
non-uniform
capacities~\cite{filtser2025fptapproximationscapacitatedsum}. For
capacitated \MSD, FPT approximations of $(4+\eps)$ and $(7+\eps)$ are
achievable for uniform and non-uniform capacities,
respectively~\cite{filtser2025fptapproximationscapacitatedsum}. FPT
approximation algorithms have also been designed for variants
incorporating \emph{fairness
  constraints}~\cite{10.1609/aaai.v38i18.30053,carta_et_al:LIPIcs.ISAAC.2024.16}. More
generally, recent works have explored FPT approximations for \MSR
under \emph{mergeable constraints}, where the clusters are required to
satisfy constraints having the property that if two disjoint clusters
satisfy the constraint, then so does their union. Carta et
al.~\cite{carta_et_al:LIPIcs.ISAAC.2024.16} gave a
$(6+\eps)$-approximation algorithm for this problem, which was
improved to a $(4+\eps)$-approximation by Bandyapadhyay and
Chen~\cite{bandyapadhyay_et_al:LIPIcs.APPROX/RANDOM.2025.23}.

\medskip\noindent\textbf{Other clustering objectives.} FPT approximation algorithms have also been extensively studied for other popular clustering objectives such as $k$-center, $k$-median and $k$-means. Here, we only discuss the results in the unconstrained setting. Via the standard reduction from dominating set, a better-than-$2$ factor approximation for the $k$-center problem is known to be $W[2]$-hard, and a $2$-approximation in polynomial time is folklore. For the $k$-median and $k$-means problems, Cohen-Addad et al.~\cite{cohenaddad_et_al:LIPIcs.ICALP.2019.42} gave FPT algorithms with approximation ratios of $1 + 2/e + \eps$ and $1 + 8/e + \eps$ respectively, and (almost) matching hardness results %
assuming gap-ETH.

\section{Preliminaries and Notation}
\label{sec:notation}

For a metric space $\cM = (X,d)$, a \emph{$k$-clustering} (or a
\emph{$k$-partition}) is a collection of subsets
$S_1, \ldots, S_{k'}$, where $k' \leq k$, with
$X = \cup_{i \in [k']} S_i$.  Observe that we allow a $k$-clustering
to have fewer than $k$ clusters. For the objective functions that we
consider, this is without loss of generality: one can transform any
clustering with $k' < k$ clusters into one having exactly $k$ clusters
and no larger cost, by simply removing $k-k'$ points from some of the
clusters and forming singleton clusters using these points.

Given $c\in X$ and a value $r\geq 0$, we let $\Ball(c,r)$ denote the
\emph{ball} of radius $r$ around the center $c$: i.e.,
$\Ball(c,r) := \{x \in X \mid d(c,x) \leq r\}$. Given a ball $B$, we
use $c(B)$ and $r(B)$ to denote its \emph{center} and \emph{radius}
respectively. Recall that the \emph{diameter} of set $S \sse X$ is defined as
$\diam(S) := \max_{p,q \in S} d(p,q)$;
any vertex pair $p,q\in S$ with
$d(p,q) = \diam(S)$ is called a \emph{diameter witness} pair.

When considering the \MSR problem, we define the cost for a set $\cB$
of balls to be $\cost(\cB):= \sum_{B \in \cB} r(B)$.
Analogously, when considering the \MSD problem, the cost of a set
$\cC$ of clusters is $\cost(\cC):= \sum_{C \in \cC} \diam(C)$.

Slightly abusing notation, we use the term \emph{cluster} and
\emph{ball} interchangeably for \MSR.  In our algorithms for \MSR, we
will consider subproblems identified by subsets $X'\subseteq X$ of
points. Having fixed some optimal solution $\OPT$ for the \MSR
problem, given $X'\subset X$, we define
$\OPT(X') :=\{\OPT_i\in \OPT \mid \OPT_i\cap X'\neq \emptyset\}$ to be
the set of optimal balls that contain at least one point in $X'$. We
denote the corresponding cost by
$\opt(X'):=\sum_{\OPT_i\in \OPT(X')}r(\OPT_i)$.

\subsection{Preprocessing for \MSR and \MSD}
\label{sec:preprocessing-msr}

Our deterministic algorithms for the \MSD and \MSR problems start by preprocessing the
instance in polynomial time. In particular, this guarantees that the pairwise distances can only take $O(k/\eps)$ values, which allows us to guess some relevant quantities in FPT time.
(The proof is deferred to~\Cref{sec:missing}.)

\begin{lemma}[Preprocessing]
  \label{lem:reduction:radii}
  Let $\eps\in (0,1]$. Given a $\rho(\eps)$-approximation algorithm
  with running time $T(n,k,\eps)$ for instances of \MSR (respectively,
  \MSD) where all the distances are positive integers,%
  and the optimum
  solution has cost at most $\frac{8k}{\eps}$ (resp.,
  $\frac{16k}{\eps}$), there exists a
  $(1+\eps)\,\rho(\eps)$-approximation algorithm for general instances
  of \MSR (resp., \MSD) with running time $T(n,k,\eps)+ n^{O(1)}$.
\end{lemma}

\section{Algorithm for Min-Sum Diameters}
\label{sec:algorithm-min-sum}

In this section, we present our algorithm \recMSD for \MSD, which is
formally presented as \Cref{alg:PTAS_MSD_rnd}. This algorithm takes as
input a subset $S$ of the point set $X$ and an integer parameter $t$,
and it returns a vector $(\cC_1, \cC_2, \ldots, \cC_t)$, where $\cC_i$
is an $i$-clustering of $S$. The algorithm first initializes each
$\cC_i$ as the na\"{\i}ve clustering that puts all of $S$ in a single
cluster; clearly having all of $S$ in one cluster is a feasible
$i$-clustering for any $i \geq 1$. Now suppose $x, y \in S$ form a
diameter witness pair for $S$, i.e., $d(x,y) = \diam(S)$. We choose a
threshold $R \in [0, \diam(S)]$ uniformly at random, and use it to
split $S$ into two sets $S_1$ and $S_2$, where
$S_1 :=S\cap \Ball(x,R)$ is the set of points in $S$ at distance at
most $R$ from $x$ and $S_2 := S \setminus S_1$ is the set of the
remaining points. We then recursively compute
$\cA \gets \recMSD(S_1, t-1)$ and $\cB \gets \recMSD(S_2, t-1)$, where
the second parameter indicates that we should return at most $t-1$
clusters. Finally, we compute their $(\min, +)$ convolution---i.e.,
for each value $i$, we find the best clustering of $S$ formed by
taking some $j$-clustering from $\cA$, and some $(i-j)$-clustering
from $\cB$.

\begin{algorithm}
  \caption{The Procedure \recMSD}
  \label{alg:PTAS_MSD_rnd}
  \DontPrintSemicolon
  \KwIn{$S\subseteq X$ with $S\neq \emptyset$, $t\in [k]$}
  \KwOut{A vector $\cC = (\cC_1, \cC_2, \ldots, \cC_t)$, where $\cC_i$ is an $i$-clustering of $S$.}
  $\cC_i \gets \{S\}$ for all $i \in \{1, 2, \ldots, t\}$\;
  \lIf{$t=1$ or $|S| = 1$}{\textbf{return} $(\cC_1, \cC_2, \ldots
    \cC_t)$}
  \tcp{Randomly partition $S$ and recurse}
  Let $x, y \in S$ be such that $d(x, y) = \diam(S)$, and choose $R\in [0,\diam(S)]$ u.a.r.\;
  Let $S_1\leftarrow S\cap \Ball(x,R)$ and $S_2\leftarrow S\setminus \Ball(x,R)$\;
  Let $\cA \gets \recMSD(S_1, t - 1)$ and $\cB \gets \recMSD(S_2, t-1)$\;
  \tcp{Perform a $(\min, +)$ convolution of $\cA$ and $\cB$}
  \For{i, j = $1, 2, \ldots t - 1$}{
    \lIf{$i+j \leq t$ and $\cost(\cA_i \cup \cB_j) <$ $\cost(\cC_{i+j})$}{
      $\cC_{i+j} \gets \cA_i \cup \cB_j$
    }
  }
  \Return $(\cC_1, \cC_2, \ldots, \cC_t)$
\end{algorithm}

The main technical lemma, which we prove next, shows that the output
clusterings have near-optimal cost with sufficiently high probability:

\begin{lemma}
    \label{lem:msd_prob}
    Let $\opt_i(S)$ denote the cost of the optimal $i$-clustering of
    $S$ and $(\cC_1, \cC_2, \ldots \cC_t)$ be the output of
    $\recMSD(S, t)$. Then, for any $\eps \in [0, 1)$ and any
    $r \in \{1, 2, \ldots t\}$,
    $$\Pr\left[\cost(\cC_r) \leq \frac{\opt_r(S)}{1 - \eps}\right] \geq \eps^{r-1}.$$
\end{lemma}
\begin{proof}
  We prove this via induction on $t$. In the base case, namely $t = 1$, the statement is trivially true since there is only one
  $1$-clustering. Otherwise, if $\opt_r(S) \geq \diam(S) (1 - \eps)$,
  then again the statement is true, since
  $\text{cost}(\cC_r) \leq \diam(S)$. Otherwise, consider the clusters
  $X_1, X_2, \ldots X_{r'}$ in the optimal $r$-clustering of $S$,
  where $r' \leq r$. For each cluster $X_h$, define $L_h$ to be the
  smallest distance from $x$ to some point in $X_h$ and $R_h$ to be
  the largest distance from $x$ to some point in $X_h$, and define
  $I_h$ to be the interval $[L_h, R_h]$. Observe that the length of
  $I_h$ is at most $\diam(X_h)$, since for any two points
  $a, b \in X_h$, $d(x, b) - d(x, a) \leq d(a, b) \leq \diam(X_h)$ by
  the triangle inequality. Hence, the sum of the lengths of the $r'$
  intervals corresponding to the optimal clusters is at most
  $\opt_r(S)$. Now, since $\opt_r(S) < (1-\eps) \diam(S)$, the
  threshold $R$ does not lie inside any of these intervals with
  probability at least $\eps$. In this event, we get a split $S_1$,
  $S_2$ where $S_1$ is precisely a union of some $r_1$ clusters of the
  optimal $r$-clustering of $S$ and $S_2$ is the union of the
  remaining $r'_2 = r' - r_1$ clusters. Also, let $r_2 = r - r_1 \geq r_2'$. The algorithm computes $\cA$ as
  the output of $\recMSD(S_1, t - 1)$ and $\cB$ as the output
  of $\recMSD(S_2, t-1)$. Since $1 \leq r_1, r_2 \le t - 1$, by the induction hypothesis, we
  get:\
    \begin{gather*}
      \Pr\left[\cost(\cA_{r_1}) \leq \frac{\opt_{r_1}(S_1)}{1 - \eps}\right] \geq \eps^{r_1-1} \quad\text{and}\quad
      \Pr\left[\cost(\cB_{r_2}) \leq \frac{\opt_{r_2}(S_2)}{1 - \eps}
      \right] \geq \eps^{r_2-1}. 
    \end{gather*}
    If both these events happen, then after the $(\min,+)$ convolution, we must have:
    \begin{gather*}
      \text{cost}(\cC_r) \leq \text{cost}(\cA_{r_1}) +
      \text{cost}(\cB_{r_2}) \leq \frac{\opt_{r_1}(S_1) +
        \opt_{r_2}(S_2)}{1 - \eps} \leq \frac{\opt_{r_1}(S_1) + \opt_{r'_2}(S_2)}{1-\eps}= \frac{\opt_r(S)}{1 - \eps}.
    \end{gather*}
    where the second inequality holds since $r_2' \le r_2$ and every $r_2'$-clustering is also a $r_2$-clustering. Now, two events are independent since $\recMSD(S_1,t-1)$
    and $\recMSD(S_2, t-1)$ are independent, and consequently
    the probability that both these events happen is at least
    $\eps^{r_1 + r_2-2} \geq \eps^{r-2}$. Therefore,
    \begin{equation*}
      \Pr\left[\cost(\cC_r) \leq \frac{\opt_r(S)}{1 - \eps}\right]
      \geq \eps \cdot \eps^{r-2} = \eps^{r-1}. \qedhere
    \end{equation*}
\end{proof}
 Finally, it remains to bound the running
time of the $\recMSD$ procedure.

\begin{lemma}
  \label{lem:MSD_time}
  The procedure $\recMSD(X, k)$ runs in time $n^{O(1)}$.
\end{lemma}

\begin{proof}
  Consider the recursion tree of $\recMSD$. We claim that any node corresponding to a call $\recMSD(S, t)$ has at most $2|S| - 1$ nodes in its sub-tree. This is trivially true for each leaf node. Also, any non-leaf node corresponding to a call $(S, t)$ has exactly two children $(S_1, t-1)$ and $(S_2, t-1)$ where $|S_1| + |S_2| = |S|$ with $1 \le |S_1|, |S_2| \le |S|-1$ since $x \in S_1, y \in S_2$. Thus, by the induction hypothesis, the total number of nodes in the subtree of this node is at most $1 + 2|S_1|-1+2|S_2|-1=2|S|-1$. Thus, there are $O(n)$ nodes in the recursion tree and we spend $n^{O(1)}$ time per node and the statement of the lemma follows.
\end{proof}

We can thus return the best $k$-clustering out of $O(1/\eps)^k \log n$
runs of $\recMSD(X, k)$ as our final output. The approximation factor
of our output is $\frac{1}{1 - \eps} = 1 + O(\eps)$ for small enough
$\eps$, with high probability, getting the following theorem:

\begin{theorem}
  For any $0 < \eps \leq 1/2$, there exists a randomized algorithm
  that runs in time $(1/\eps)^k \cdot n^{O(1)}$, and
  outputs a $(1 + 2 \eps)$-approximate solution for \MSD with high probability.
\end{theorem}

Substituting $\eps=1$ in the algorithm
from~\cite{doi:10.1137/1.9781611977912.69} gives an $8$-approximation
for \MSD in time $n^{O(1)}$; we can take the best of our output and
this $8$-approximation to get an $(1+\eps)$-approximation in
expectation while incurring an additional additive factor of
$n^{O(1)}$ in the runtime.  %
Finally, we can derandomize the above algorithm in time
$(k/\eps)^k \cdot n^{O(1)}$ by first preprocessing the input to have
integral distances in $[1, O(k/\eps)]$, and then trying every integral
possibility of the threshold $R$ in $[0, \diam(S))$ instead of
choosing a random one, thus proving~\Cref{thr:mainMSD}; we defer the
details to \Cref{sec:derandomization}.

Moreover, the ideas above can be used %
to obtain a simple deterministic $n^{k + O(1)}$ time exact algorithm for \MSD, improving and simplifying the previous
$n^{5k + O(1)}$ algorithm from \cite{10.1609/aaai.v39i15.33699}. We
defer the details to~\Cref{sec:exact_msd}. The same algorithm can also
be utilized to obtain a $(1+\eps)$ approximation for the \MSD problem
(and thus, a $(2+\eps)$-approximation for the \MSR problem) under any
mergeable constraints in FPT time. The details for the same are
deferred to~\Cref{sec:mergeable_clustering}.

\section{An FPT Approximation Scheme for \MSR}
\label{sec:fpt_apx_scheme_msr}

In this section we describe our FPT approximation scheme for
\MSR. Using the preprocessing lemma (\Cref{lem:reduction:radii}), we
can assume that the distances are integers and the optimum is bounded
by $\frac{8k}{\eps}$. We fix a total ordering $\prec$ over $X$, and
assume that $\prec$ is globally available throughout the run of our
algorithm. Let $\OPT$ denote the optimal solution (interpreted as a set of balls). We will assume that
no optimal ball in $\OPT$ contains the center of another optimal ball
(since otherwise, we could expand the first optimal ball by the radius
of the second optimal ball to include the second optimal ball inside
it, without adding to the cost).

\subsection{The General Approach}
\label{sec:general_approach_msr}
Our algorithm is recursive in nature. Suppose that we need to cluster
a set $S\subseteq X$ of points using at most $t\leq k$
clusters. Recall from~\Cref{sec:notation} that $\OPT(S)$ denotes the
balls of the optimal solution $\OPT$ (for the original problem on
$(X,d)$ with up to $k$ clusters) which contain at least one point from
$S$. Suppose that $\OPT(S)$ consists of balls
$B_1^\star = \Ball(c_1^\star, r_1^\star), \ldots, B_h^\star =
\Ball(c_h^\star, r_h^\star)$, where $c_i^\star$ and
$r_i^\star$ denote the center and radius of the optimal ball
$B_i^\star$. We will assume that $h \le t$ and $c_i^\star \in S$ for all $i \in [h]$. We call pairs $(S, t)$ that satisfy this condition as \emph{valid pairs}.

\begin{defn}[Valid Pair]
  For a subset $S \subseteq X$ and a parameter $t \le k$, the pair
  $(S, t)$ is said to be \emph{valid}, if $|\OPT(S)| \le t$ and the center of
  each ball in $\OPT(S)$ belongs to $S$.
\end{defn}

Although not all subproblems we encounter will form a valid pair, our analysis will recursively maintain the invariant that the input to each recursive call is a valid pair. In particular, observe that the original input $(X, k)$ forms a valid pair. Our algorithm consists of three phases:
\begin{enumerate}
\item \emph{Computing a ``defensive clustering'':} In this phase, we cluster
  the subset $S$ of points into ``defensive balls'', such that each
  optimal ball in $\OPT(S)$ is contained in at least one defensive
  ball. We ensure that these defensive balls themselves do not have
  high cost, and satisfy some other good properties. (See
  \Cref{sec:defensive-clustering} for the definition of the defensive
  clustering, and \Cref{sec:constr-defensive-clust} for the algorithm
  to construct them.)

\item \emph{Defining the subproblems:} In this phase, we extend a ball
  $E$ (initialized as an arbitrary defensive ball), until we can
  divide the problem into three parts: the points in the defensive
  balls on the boundary of $E$, and then the non-boundary points
  decompose into those falling inside $E$ and those outside $E$. At
  the end of this extension, we want that the defensive balls on the
  boundary have a negligible cost compared to the radius of $E$, and
  that a constant number of optimal balls inside $E$ account for a
  constant fraction of the radius of $E$. This is done in \Cref{sec:expansion-process}.

\item \emph{Recursively Solving the Subproblems:} Since the defensive
  balls on the boundary in our solution are inexpensive, we can just
  include them in our solution (for technical reasons, we might need to expand these balls before including them, while charging the cost of expansion to the additional covered optimal balls). We now cover the points not covered by
  these balls (the ``non-boundary'' points) by recursively solving the
  problems inside and outside $E$ independently. However, for the
  inside problem, we need to first guess a constant number $M$ of
  optimal balls before the recursive call, thereby reducing the
  optimal value on the sub-instance, and ensuring rapid progress.
\end{enumerate}

\subsection{Phase 0: Defining the ``Ideal'' Defensive Clustering}
\label{sec:defensive-clustering}

To state our algorithm, we define the ``defensive clustering''---this
is the clustering which we find in the first phase. Hence the
algorithm presented in this section is really a thought experiment,
which takes the optimal solution in addition to a subset $S \sse X$,
and outputs a collection $\{D_1, D_2, \ldots, D_{t'}\}$ of
\emph{defensive balls}. This collection of balls is said to be a
\emph{defensive clustering} of $\OPT(S)$ if:
\begin{enumerate}[label=(\roman*)]
\item For each \emph{defensive ball} $D_j$, there exists a non-empty
  subset $I(j) \subseteq [h]$, such that for all
  $i \in I(j), B_i^\star \subseteq D_j$. For each $i \in I(j)$, we say
  that the ball $D_j$ \emph{defends} $B_i^\star$ and that $D_j$ is the
  \emph{defender} of $B_i^\star$.
\item The sets $I(1), I(2), \ldots, I(t')$ form a partition of $[h]$. Crucially, note that an optimal ball might be inside multiple defensive balls. However, it will still have exactly one designated defender.
\item For each defensive ball $D_j$, there exists an index
  $\kappa(j) \in I(j)$, such that the radius $r_j$ of $D_j$ satisfies
  $r_j \le r_{\kappa(j)}^\star + \sum_{i \in I(j)} r_i^\star$. The
  ball $B_{\kappa(j)}^\star$ is designated as the \emph{core} of
  $D_j$.
\end{enumerate}

In this section, we show that such a clustering indeed exists; we then
algorithmically find such a clustering in
\Cref{sec:constr-defensive-clust}.  Note that any defensive clustering
is already a $2$-approximate solution, since
$$\sum_{j=1}^{t'} r_j = \sum_{j=1}^{t'} r_{\kappa(j)}^\star +
\sum_{j=1}^{t'} \sum_{i \in I(j)} r_i^\star \leq 2 \sum_{i=1}^{h}
r_i^\star.$$

In the thought experiment formalized as {\tt
  DefensiveClustering} in \Cref{alg:compute_defensive_cluster}, we
imagine being given $\OPT(S)$ as input. Initially, every optimal ball
is marked undefended, and we have no defensive balls. Suppose at some
step, we have found $j-1$ defensive balls, but they do not yet cover
$S$. We open a new defensive ball $D_j$ as follows: we choose the
smallest uncovered point (according to the ordering $\prec$) to be the
center $c_j$ of $D_j$, and initialize its radius $r_j$ to be twice the
radius of the least-indexed optimal ball (say $B_i^\star$) containing
$c_j$---this ensures that $B_i^\star \subseteq D_j = \Ball(c_j,
r_j)$. We mark every undefended optimal ball---and, in particular, the
ball $B_i^\star$---which lies inside $D_j$ as being defended by $D_j$,
and let $I(j)$ be the set of balls defended by $D_j$. We also
designate the optimal ball $B_i^\star$ as being the \emph{core} of
defensive ball $D_j$ by setting $\kappa(j) = i$.

\bigskip
\begin{algorithm}[H]
  \caption{{\tt DefensiveClustering}}
  \label{alg:compute_defensive_cluster}
  \DontPrintSemicolon
  \textbf{Global:} The metric space $(X, d)$, the parameter $k$, and a total order $\prec$ on $X$.\;
  \KwIn{An optimal $t$-clustering $\OPT(S) = \{B_1^\star, \ldots, B_h^\star\}$ of a subset $S$ of $X$ such that $B_i^\star = \Ball(c_i^\star, r_i^\star)$ for all $i\in[h]$.}
  \KwOut{A defensive clustering of $\OPT(S)$.}
  $j \gets 1$\;
  $U \gets S$\tcp*{uncovered points}
  $Z \gets [h]$ \tcp*{indices of undefended optimal balls}
  \While{$U \neq \emptyset$}{
    $c_j \gets $ smallest point of $U$ according to $\prec$\tcp*{center of the new defensive ball}
    $i \gets \min \{u \in [h]: c_j \in B_u^\star\}$\tcp*{least index among optimal balls containing $c_j$}
    $r_j \gets 2 r_i^\star$ \tcp*{radius of the new defensive ball}
    $\kappa(j) \gets i$\tcp*{index of the core of the new defensive ball}
    $I(j) \gets \{z \in Z: B_z^\star \subseteq \Ball(c_j, r_j)\}$\tcp*{balls marked as defended by $\Ball(c_j, r_j)$}
    $Z \gets Z \setminus I(j)$\;
    \While{$\exists$ an $i' \in Z$, such that $c_{i'}^\star \in \Ball(c_j, r_j)$}{
      Let $i''$ be the index of the largest such optimal ball\;
      $r_j \gets r_j + r_{i''}^\star$\;
      $I(j) \gets I(j) \cup \{z \in Z: B_z^\star \subseteq \Ball(c_j, r_j)\}$\;
      $Z \gets Z \setminus I(j)$\;
    }
    $D_j \gets \Ball(c_j, r_j)$, $U \gets U \setminus D_j$, $k_j \gets |I(j)|$\;
    $j \gets j + 1$\;
  }
  \Return $\{D_1, D_2, \ldots, D_{t'}\}, (k_1, k_2, \ldots, k_{t'})$ for
    $t' \gets j - 1$\;
\end{algorithm}
\bigskip

Now, while the defensive ball $D_j$ contains the center of some
undefended optimal ball, we increase the radius $r_j$ by the radius of
the highest-radius such optimal ball (say $B_{i''}^\star$). We say that
each undefended optimal ball which gets included inside $D_j$ after
this expansion as being defended by $D_j$. (In particular, the
largest-radius ball $B_{i''}^\star$ is marked as being defended by
$D_j$, and $i'' \in I(j)$.) \Cref{fig:defensive_ball_image} shows the
creation process for a defensive ball. The process stops when the ball
$D_j$ does not contain the center of any undefended optimal ball.

We now verify that the output of
{\tt DefensiveClustering}
satisfies all the requirements of a defensive clustering. Indeed, for
any optimal ball $B_q^\star$, consider the first moment when its
center $c_q^\star$ is contained in some defensive ball. This could be
a result of one of the following two events: (a)~the opening of a new
defensive ball $D_j$, or (b)~the expansion of the latest defensive
ball $D_j$. In either case, just before this event, $c_q^\star$ was
uncovered, and hence $B_q^\star$ was undefended---equivalently,
$q \in Z$ in the pseudo-code. After this event, either
$B^\star_q \subseteq D_j$, in which case it will get marked as
defended by $D_j$ immediately, or we will expand $D_j$ to fully
include $B_q^\star$ inside $D_j$ in the next step, and then mark
$B_q^\star$ as defended by $D_j$. Hence, at termination, every optimal
ball is defended by exactly one defensive ball, and
$I(1), \ldots, I(t')$ form a partition of $[h]$.

Also, when a defensive ball $D_j$ is first opened and $B_i^\star$ is
an optimal ball covering $c_j$, the radius $r_j$ is initialized to be
$2r^\star_{i}$, the index $i$ is included into $I(j)$ and the index of
$D_j$'s core optimal ball $\kappa(j)$ is set as $i$. Later, whenever we expand
the ball $D_j$ by the radius $r_{i''}^\star$ of some optimal ball
$B_{i''}^\star$, we include $i''$ in $I(j)$. Hence,
$r_j \le 2r_{\kappa(j)}^\star + \sum_{i \in I(j)\setminus
  \{\kappa(j)\}} r_i^\star = r_{\kappa(j)}^\star + \sum_{i \in I(j)}
r_i^\star$.

\subsection{Phase 1: Algorithmically Finding the Defensive Clustering}
\label{sec:constr-defensive-clust}

Having defined defensive clusterings via the thought experiment given
in \texttt{DefensiveClustering}, we now show how to find it in FPT
time. Note that, for any fixed choice of
$(t', r_1, r_2, \ldots, r_{t'})$, the center $c_1$ must be the smallest
point in $S$ (according to $\prec$), and in general, the center $c_i$
must be the smallest point in
$S \setminus (\cup_{j < i} \Ball(c_j, r_j))$.  Thus, the sequence of
the centers is uniquely determined, and hence so are the defensive
balls. Since the defensive clustering is a $2$-approximate solution
and the optimum is bounded by $\frac{8k}{\eps}$, we must have:
$$\sum_{j=1}^{t'} r_j \le \frac{16k}{\eps}.$$

Apart from the set of the defensive balls, our algorithm will need the
values $k_j = |I(j)|$ for all $j \in [t']$.
Thus, going ahead, we will slightly abuse notation to assume that
``defensive clustering'' denotes the set of defensive balls, along
with the sequence $k_1, k_2, \ldots, k_{t'}$.

We now present an algorithm \texttt{CandidateClusterings},
whose pseudocode appears as
\Cref{alg:defensive_cluster_candidates}. This algorithm iterates over
all valid choices of
$(t', (r_1, \ldots, r_{t'}), (k_1, \ldots, k_{t'}))$, and returns all
the valid clusterings corresponding to these choices. The following
\Cref{lem:small_set} states that the number of output clusterings is
small and at least one of them must be a defensive clustering of
$\OPT(S)$.

\begin{algorithm}[H]
  \caption{{\tt CandidateClusterings}}
  \label{alg:defensive_cluster_candidates}
  \DontPrintSemicolon
  \textbf{Global:} The metric space $(X, d)$, the parameter $k$, and a total order $\prec$ on $X$.\;
  \KwIn{subset $S\subseteq X$ with $S\neq \emptyset$, and parameter $1 \le t \le k$.}
  \KwOut{A set of defensive clusterings of $S$}
  $\cC \gets \emptyset$\;
  \For{$t' \in \{1, 2, \ldots, t\}$}{
    \For{each tuple $(r_1, r_2, \ldots, r_{t'})$ of non-negative integers with sum $\le \frac{16k}{\eps}$}{
      error $\gets \text{false}$\;
      \For{$i = 1, 2, \ldots, t'$}{
        $c_i \gets$ the smallest element of $S \setminus \cup_{j=1}^{i-1} \Ball(c_j, r_j)$ according to $\prec$\;
        \If{no such element exists}{
          error $\gets \text{true}$\;
          \textbf{break}\;
        }
        $D_i \gets \Ball(c_i, r_i)$
      }
      \If{error = false and $\cup_{j=1}^{t'} \Ball(c_j, r_j) = S$}{
        \For{each tuple $(k_1, k_2, \ldots, k_{t'})$ of positive integers with sum $\le t$}{
          $C \gets (\{D_1, D_2, \ldots, D_{t'}\}, (k_1, k_2, \ldots, k_{t'}))$\label{l:computed_defensive_clustering}\;
          $\cC \gets \cC \cup \{C\}$\;
        }
      }
    }
    
  }
  \Return $\cC$\;
\end{algorithm}

\begin{lemma}
  \label{lem:small_set}
  {\tt CandidateClusterings} %
  returns a set $\cC$ of
  clusterings with size at most $(\frac{1}{\eps})^{O(k)}$, such that
  $\DefensiveClustering(\OPT(S)) \in \cC$.
\end{lemma}

\begin{proof}
  The clustering $C$ computed in
  line~\ref{l:computed_defensive_clustering} with the choice of
  $(t', r_1, \ldots, r_{t'}, k_1, \ldots, k_{t'})$ borrowed from
  $\DefensiveClustering(\OPT(S))$, is exactly
  $\DefensiveClustering(\OPT(S))$. Now the number of tuples of $t'$
  non negative integers with sum not exceeding $s$ equals
  $\binom{s+t'}{s}$. Thus, by the hockey stick identity, the number of
  tuples of at most $t$ non-negative integers with sum not exceeding
  $s$ equals $\binom{s+t+1}{s+1}$. Substituting
  $s \le \frac{16k}{\eps}$ and $t \le k$ and using
  $\binom{n}{k} \leq (\frac{en}{k})^k$, we get that the number of
  possibilities for $(t', r_1, \ldots, r_{t'})$ is at most:
  \begin{gather}
    \binom{\lfloor 16k/\eps \rfloor + k + 1}{k + 1} \leq \left(\frac{e
        (16k/\eps + k + 1)}{k + 1} \right)^{k+1} \leq
    \left(\frac{1}{\eps}\right)^{O(k)}.
  \end{gather}
  The number of tuples $(k_1, k_2, \dots, k_{t'})$ of positive integers
  having sum at most $t$ is upper bounded by $2^t \le 2^k$, since the
  number of tuples of positive integers with sum $m$ is given by
  $2^{m-1}$. Hence, we get $|\cC| \le (\frac{1}{\eps})^{O(k)}$.
\end{proof}

\subsection{The Recursive Algorithm}
\label{sec:recursive-algorithm}

Our algorithm uses a recursive procedure \recMSR (presented in
\Cref{alg:rec_msr}), which uses constants $L = \ceil{\frac{6}{\eps} \ln \frac{1}{\eps}}$ and $M = \ceil{\frac{1}{\eps^2}}$. The
procedure takes as input a set $S\subseteq X$ of uncovered points, an
upper bound $t$ on the number of clusters and a depth parameter
$0 \le \ell \le L$ which we call the \emph{level} of the
subproblem. It returns a $t$-clustering of $S$. The output for the original input problem shall be computed by
$\texttt{recMSR}(X, k, 0)$.
The procedure
  is not allowed to open centers outside $S$. Note that, in general
  it might be possible that for a cluster $C \subseteq S$ of points,
  the optimal center is outside the subset $S$. However, we will
  design our algorithm in such a way that it suffices for centers to
  be inside the input set of points for the recursive subproblems. We enforce this requirement since the correctness of our thought experiment {\tt DefensiveClustering}(\Cref{alg:compute_defensive_cluster}) relies on the assumption that the centers of all the optimal balls that intersect $S$ are in $S$ itself. The correctness of~\Cref{alg:compute_defensive_cluster} is in-turn crucial to get an FPT number of candidate defensive clusterings via~\Cref{alg:defensive_cluster_candidates}.

\begin{algorithm}[h]
  \caption{\recMSR procedure}
  \label{alg:rec_msr}
  \DontPrintSemicolon
  \textbf{Global:} The metric space $(X, d)$, the parameters $k, M, L$\;
  \KwIn{A subset $S\subseteq X$, $t \in [k]$, $\ell \in \{0,\ldots,L\}$}
  \KwOut{A clustering of $S$ with at most $t$ balls.}
  \If{$S = \emptyset$}{\Return $\emptyset$\tcp*{no clusters required, $\cost(\emptyset) = 0$}}
  \If{$t = 0$}{
    \Return null \tcp*{no way to cluster, $\cost(null) = \infty$}
  }
  $\cC \gets {\tt CandidateClusterings}(S, t)$\;

  $\cT \gets \cC$ \tcp*{candidate solutions}
  \lIf{$\ell = L$}{
    \Return the best solution from $\cT$
  }
  \For{\label{l:clustering_loop}each clustering $(\{D_1, D_2, \ldots, D_{t'}\}, (k_1, k_2, \ldots, k_{t'})) \in \cC$}{
    \For{$j=1,2,\ldots, t'$}{
      \For{$r'_j = 0, 1, \ldots, \lfloor \frac{8k}{\eps} \rfloor$\label{l:expansion_1}}{
        $D'_j \gets \Ball(c(D_{j}),r(D_j) + r'_j)$ \label{l:good_ball_expansion}\tcp*{expanded good defensive ball}
        $\cA \gets \recMSR(S \setminus D'_j, t - 1, \ell)$\label{l:insignificant_core_recursion}\;
        \If{$\cA$ is not null}{
        $\cT \gets \cT \cup \{\{D'_j\} \cup \cA\}$\tcp*{a candidate solution}
        }
      }
    }
    $E, J_b, k_{in}, k_{out} \gets \texttt{expandedBall}(S, (\{D_1, \ldots, D_{t'}\}, (k_1, \ldots,k_{t'})))$\;
    Let $J_b$ be $\{j_1, j_2, \ldots, j_s\}$\;
    \For{tuples $(r'_1, r'_2, \ldots, r'_{s})$ of non negative integers with sum $\leq \lfloor \frac{8k}{\eps}\rfloor$ \label{l:expansion_2}}{
      $D'_q \gets \Ball(c(D_{j_q}),r(D_{j_q}) + r'_{q})$ for all $q \in [s]$\label{l:expanded_boundary_ball}\tcp*{expanded boundary defensive balls}
      $S_{out} \gets S \setminus (E \cup (\cup_{q \in [s]} D'_q))$ \label{l:s_out}\tcp*{outer uncovered points}
      $S_{in} \gets S \cap (E \setminus (\cup_{q \in [s]} D'_q))$\label{l:s_in}\tcp*{inner uncovered points}
      $\cB^\star_{out} \gets \recMSR(S_{out}, k_{out}, \ell)$ \label{l:outside_recursion} \tcp*{direct recursive call for outer points}
      $\cB^\star_{in} \gets {\tt guessAndSolveMSR}(S_{in}, k_{in}, \ell + 1)$ \tcp*{guess and solve for inner points}
      \If{$\cB^\star_{in}$ and $\cB^\star_{out}$ are not null}{
      $\cT \gets \cT \cup \{\{D'_1, \ldots, D'_s\} \cup \cB_{in}^\star \cup \cB^\star_{out}\}$ \tcp*{a candidate solution}
      }
    }
  }
  \Return the best solution from $\cT$
\end{algorithm}

The high level idea of the procedure is to choose a set $\cB$ of balls
to include in our solution, and cover the points not covered by $\cB$
recursively by dividing the remaining problem into subproblems. The
challenge is to ensure that the subproblems do not interact
\emph{badly} with one another, while also making progress at every
level of the recursion tree. Let us first describe a key property that we would like the set of our chosen
balls to have:
\begin{defn}[Central property]
    A set $\cB$ of balls is said to satisfy the \emph{central property} if for each optimal ball $B^\star_i$ whose center $c_i^\star$ is covered by $\cB$, $\cB$ must cover all of $B_i^\star$, that is$$c_i^\star \in \cup_{B \in \cB} B \implies B_i^\star \subseteq \cup_{B \in \cB} B.$$
\end{defn}
Observe that, if a set $\cB$ of balls satisfying the central property is removed from $S$ (where $(S, t)$ forms a valid pair), then the remaining set $S'$ of uncovered points satisfies one of the two requirements of a valid pair: for each optimal ball $B_i^\star$ that intersects $S'$, its center $c_i^\star$ must be in $S'$. To understand this requirement, suppose this property was not satisfied. Then, there might be an optimal ball, say
$B_i^\star$ whose center $c_i^\star$ is covered by $\cB$, but (in the extreme case) none of the other points in $B_i^\star$ are covered. Hence, in the recursive subproblems obtained on removing the set of points covered by $\cB$, $c_i^\star$ will not be present. Suppose there exists a subproblem with the set of points being $B_i^\star \setminus \{c_i^\star\}$ and the number of allowed clusters being $1$. In this case, the optimal cost of covering $B_i^\star \setminus \{c_i^\star\}$ by a ball centered at some point in $B_i^\star \setminus \{c_i^\star\}$ might be twice as much as $r_i^\star$.
The following lemma states that for any set $\cB$ of defensive balls, we can modify them to satisfy the central property by expanding them one by one and charging the expansion to the additional optimal balls covered.
\begin{lemma}
\label{lem:expansion_for_central_property}
    Consider any set $\cB = \{D_{j_1}, D_{j_2}, \ldots, D_{j_s}\}$ of defensive balls corresponding to indices $J = \{j_1, j_2, \ldots, j_s\}$. In time $(\frac{1}{\eps})^{O(k)} \cdot n^{O(1)}$, one can compute a collection of $(\frac{1}{\eps})^{O(k)}$ sets of balls such that at least one set $\cB'$ in the collection must satisfy the central property along with:$$\cost(\cB') \leq \sum_{j \in J} r_{\kappa(j)}^\star + \sum_{i \in T}r_i^\star$$ where $T$ denotes the indices of the set of optimal balls fully covered by $\cB'$. 
\end{lemma}
\begin{proof}
Consider the following thought experiment. Let us initialize $\cB'$ as $\cB$ and update $\cB'$ as follows in order to satisfy the central property: while there exists an optimal ball $B_i^\star$ not fully covered by the union of $\cB'$ whose center $c_i^\star$ is inside some ball $Y \in \cB'$, we expand $Y$ by $r_i^\star$. The additional cost ($r_i^\star$) paid for this step can be charged to the optimal ball $B_i^\star$. In other words, the cost of the expansion is upper bounded by the sum of the radii of the optimal balls that are covered at the end by $\cB'$ but were not covered at the beginning by $\cB$. The balls in $\cup_{j \in J} I(j)$ do not contribute to the cost of expansion since they were already covered by $\cB$. Thus, we get:%
\begin{align*}
    \cost(\cB') &\le \cost(\cB) + \sum_{i \in T \setminus \cup_{j \in J} I(j)} r_i^\star\\
    &= \sum_{j \in J} r_j + \sum_{i \in T} r_i^\star - \sum_{j \in J} \sum_{i \in I(j)} r_i^\star\\
    &\leq \sum_{j \in J} r_{\kappa(j)}^\star + \sum_{i \in T} r_i^\star.
\end{align*}
where the second inequality uses $r_j \leq r_{\kappa(j)}^\star +
\sum_{i \in I(j)} r_i^\star$. Now, since the optimal solution is
unknown, we do not know the outcome $\cB'$ of this thought
experiment. However, the sum of optimal radii is at most
$\frac{8k}{\eps}$ and we can simply iterate over all
$O(\frac{1}{\eps})^k$ possible choices of the final radii of the balls
in $\cB'$ (see lines~\ref{l:expansion_1} and~\ref{l:expansion_2} of
\recMSR).
\end{proof}

We begin with a set $\cC$ of defensive clusterings computed
using ${\tt CandidateClusterings}(S, t)$, and iterate over all
choices in $\cC$. By~\Cref{lem:small_set}, there exists
$(\{D_1, D_2, \ldots, D_{t'}\}, (k_1, k_2, \ldots, k_{t'})) \in \cC$
that is a defensive clustering of $\OPT(S)$. Recall that we use
$\kappa(j)$ to denote the core of the $j^{\text{th}}$ defensive ball
$D_j$, $I(j)$ to denote the set of indices of the optimal balls
defended by $D_j$, and $k_j$ equals $|I(j)|$ for all $j \in
[t']$. Now, there are two cases:

\begin{enumerate}
\item There exists a defensive ball $D_j$ with an insignificant core,
  that is $r_{\kappa(j)}^\star \leq \eps \sum_{i \in
    I(j)}r_i^\star$. In this case, note that $D_j$ can be included in
  the solution to cover all the optimal balls with indices in $I(j)$,
  at a cost of $(1 + \eps)$ times the sum of their radii since:
  $$r_j \le r_{\kappa(j)}^\star + \sum_{i \in I(j)} r_i^\star \le (1 + \eps) \sum_{i \in I(j)}r_i^\star.$$
  Thus, we simply include $D_j$ in our solution (after expanding it if
  required to satisfy the central property) and recurse on the
  remaining points at the same level $\ell$ (line
  \ref{l:insignificant_core_recursion} of \recMSR).

\item Each defensive ball has a significant core:
  $r_{\kappa(j)}^\star > \eps \sum_{i \in I(j)} r_i^\star$ for all
  $j \in [t']$. In this case, we initialize a ball $E$ as an arbitrary
  defensive ball of non zero radius (assumed to be $D_1$ without loss
  of generality). Define the boundary $\cB$ of $E$ as the set of
  defensive balls that intersect $E$, but are not contained within
  $E$. Each defensive ball that belongs to $\cB$ is referred to as a
  boundary ball. Our plan is to expand $E$ till we can identify the
  recursive subproblems. Thus, we call $E$ the ``expanded
  ball''. %
\end{enumerate}

We now describe the expansion process that we use for obtaining this
expanded ball.

\subsubsection{Phase 2: The Expansion Process}
\label{sec:expansion-process}

In this section, we give the details of the expansion process (which
is given in \texttt{expandedBall}, see \Cref{alg:expanded_ball}), and
note some crucial properties of the expanded ball $E$ obtained at the
end of this process. Let $J_b$ denote the set of indices of all the
boundary balls, $r_{\max}$ denote the maximum radii among them, and
$r_{\text{sum}}$ denote the sum of their radii. Observe that, each boundary ball must have a non zero radius. Indeed, any ball with a zero radius must either be fully inside or fully outside $E$. Suppose the set of boundary balls is non-empty. An expansion of $E$ by $2 r_{\max}$ we will get all the boundary balls included inside $E$. We will perform this expansion by $2r_{\max}$ (line \ref{l:expand_E}), if
one of the following two cases hold:
\begin{enumerate}
\item $r_{\text{sum}}$ is at least $8 r_{\max}$: In this case,
  the expansion is clearly favorable since the sum of radii of the
  optimal balls defended by the boundary balls is at least
  $r_{\text{sum}}/2$, and we pay $r_{\max} \le r_{\text{sum}}/2$ for
  this expansion. We call such an expansion a \emph{cheap expansion}.
\item At least one boundary ball is ``large'':
  $r_{\max} \ge \eps^2
  r(E)$. %
  We call such an expansion a \emph{large expansion}. Additionally, we
  also think of the initial formation of $E$ (as $D_1$) a large
  expansion (from radius $0$ to $r_1$).
\end{enumerate}

\begin{algorithm}[h]
  \caption{{\tt expandedBall}}
  \label{alg:expanded_ball}
  \DontPrintSemicolon
  \textbf{Global:} The metric space $(X, d)$, the parameters $k, M, L$.\;
  \KwIn{non empty subset $S\subseteq X$, and a defensive clustering $(\{D_1, D_2, \ldots, D_{t'}\}, (k_1, k_2, \ldots, k_{t'}))$.}
  \KwOut{An expanded ball that cannot be expanded further.}
  $E \gets D_1$, $c\gets c(D_1)$, $r \gets r(D_1)$ \tcp*{the current expanded ball}
  $J_{in} \gets \{j \in [t']: D_j \subseteq E\}$ \tcp*{defensive balls fully inside $E$}
  $J_b\gets \{j \in [t']\setminus J_{in}: D_j \cap E \neq \emptyset\}$\ \tcp*{indices of boundary defensive balls}
  \While{$J_{b} \neq \emptyset$}{
    $r_{\text{sum}} \gets \sum_{j \in J_b} r_j$\;
    $r_{\max} \gets \max_{j \in J_b} r_j$\;
    \If{$r_{\text{sum}} \ge 8 r_{\max}$ or $r_{\max} \ge \eps^2 \cdot r(E)$ \label{l:cheap_or_large_expantion}}{
      $r\gets r+2r_{\max}$, $E\gets \Ball(c,r)$ \label{l:expand_E}\tcp*{cheap expansion or large expansion}
      $J_{in} \gets \{j \in [t']: D_j \subseteq E\}$\;
      $J_b\gets \{j \in [t']\setminus J_{in}: D_j \cap E \neq \emptyset\}$\;
    } 
    \Else{
      \textbf{break}\;
    }
  }
  $J_{out} \gets [t'] \setminus (J_{in} \cup J_b)$\tcp*{defensive balls fully outside $E$}
  $k_{in} \gets \sum_{j \in J_{in}} k_j\label{l:k_in}$\tcp*{(lower bound on) \# of optimal balls fully inside $E$}
  $k_{out} \gets \sum_{j \in J_{out}} k_j\label{l:k_out}$\tcp*{(lower bound on) \# of
    optimal balls fully outside $E$}
  \Return $E,J_b, k_{in}, k_{out}$
\end{algorithm}

The expansion process terminates when either the set $J_b$ of boundary balls is empty, or when neither cheap nor large expansion is possible. At termination, we use $J_{in}$ to denote the set of (indices of) defensive balls completely inside $E$ and $J_{out}$ to be the set of (indices of) defensive balls completely outside $E$. $k_{in} = \sum_{j \in J_{in}} k_j$ (line \ref{l:k_in}) denotes the number of optimal balls defended by the defensive balls with indices in $J_{in}$. Similarly, $k_{out} = \sum_{j \in J_{out}} k_j$ (line \ref{l:k_out}) denotes the number of optimal balls defended by the defensive balls with indices in $J_{out}$. Before we describe how the subproblems are formed and solved, we note some crucial properties of the expanded ball. We begin with bounding the sum of the radii of the core optimal balls of the boundary balls:
\begin{lemma}
\label{lem:small_boundary_core}
$\sum_{j \in J_b} r_{\kappa(j)}^\star \le 4\eps^2 r(E)$.
\end{lemma}
\begin{proof}
    If $J_b = \emptyset$, this is trivially true. If $J_b \neq \emptyset$, since no further (cheap or large) expansion of $E$ was possible, it must be that $r_{\text{sum}} < 8 r_{\max}$ and $r_{\max} < \eps^2 r(E)$. Thus, we get:
    $$\sum_{j \in J_b} r_j = r_{\text{sum}} < 8 r_{\max} < 8 \eps^2 r(E).$$
    The statement of the lemma then follows from $r_j \geq 2 r_{\kappa(j)}^\star$ since each defensive ball is initialized to have a radius twice that of its core optimal ball.
\end{proof}
Now, let $R_{\text{cheap}}$ denote the sum of all the expansions of $E$
that were cheap (that is, the sum of $2 r_{\max}$ in
line~\ref{l:cheap_or_large_expantion} when
$r_{\text{sum}} \ge 8 r_{\max}$), and let
$R_{\text{large}} = r(E) - R_{\text{cheap}}$ be the sum of all
the expansions that were large (including the initial formation of
$E$). The following lemma states that at least a constant fraction of the expansion occurs via large expansions.

\begin{lemma}
  \label{lem:mostly_large_expansions}
  $R_{\text{large}} \geq \frac{r(E)}{2}$.
\end{lemma}
\begin{proof}
  Consider a cheap expansion of $E$. This happens when
  $r_{\text{sum}} \geq 8 r_{\max}$. Note that
  $$r_{\text{sum}} = \sum_{j \in J_b} r_j \leq 2 \sum_{j \in J_b} \sum_{i \in I(j)} r_i^\star.$$
  since each defensive ball is a $2$-approximate ball for the balls it
  defends. At this point, $R_{\text{cheap}}$ gets incremented by
  $$2 r_{\max} \le \frac{r_{\text{sum}}}{4} \le \frac{1}{2} \sum_{j \in J_b} \sum_{i \in I(j)} r_i^\star.$$
  Also, after the expansion, all the optimal balls with indices in
  $\cup_{j \in J_b} I(j)$ are contained completely inside $E$. Hence,
  eventually we must have:
  $$R_{\text{cheap}} \leq \frac{1}{2} \sum_{i: B^\star_i \subseteq E} r_i^\star \leq \frac{1}{2}r(E).$$
  The final inequality follows from the optimality of $\OPT$. (If the
  sum of radii of optimal balls completely inside $E$ exceed the
  radius of $E$, then they can be replaced by $E$, thereby lowering
  the cost of an optimal solution.)
\end{proof}

\begin{lemma}
\label{lem:constant_balls_cover_constant_fraction}
  Suppose $\eps$ is a small enough positive constant and each
  defensive ball has a significant core
  ($r_{\kappa(j)}^\star > \eps \sum_{i \in I(j)} r_i^\star$) and let $M = \ceil{\frac{1}{\eps^2}}$. There exists a set of $M' = \min(k_{in}, M)$ optimal balls within $E$ whose defender balls have index in $J_{in}$, such that the sum of their radii is at least $\frac{\eps}{6} r(E)$.
\end{lemma}

\begin{proof}
  First, observe that the number of large expansions (say $M''$) must be upper bounded by $k_{in}$ since each expansion leads to the inclusion of at least one new defensive ball inside $E$. We consider two cases:
  \begin{enumerate}
      \item $M'' > M$. Thus, $M < k_{in}$ and $\min(M, k_{in}) = M$. Consider the last
  $M$ large expansions of $E$. Let $R_0$ be the radius of $E$ just
  before the last $M$ large expansions. Observe that $R_0 > 0$
    since the
  first large expansion is considered to be the formation of $E$ as
  $D_1$ itself, which must be of a non-zero radius since its core is
  significant (strictly greater than $\eps$ times the sum of the radii
  of the optimal balls it defends).
  Each
  large expansion increases the radius of $E$ by at least a factor of
  $(1 + 2\eps^2)$. Each of these expansions is done by an amount
  $2r_{\max}$, where $r_{\max}$ is the radius of the largest defensive
  boundary ball at that point in time. Consider the $M$ corresponding
  largest boundary balls. The radius of the $i^\text{th}$ of them must
  be at least $\eps^2R_0 (1 + 2 \eps^2)^{i-1}$. Hence, the sum $R'$ of
  their radii is at least:
  $$R' \ge \eps^2 R_0 \sum_{i=0}^{M-1} (1+2\eps^2)^i = \frac{R_0}{2} ((1+2\eps^2)^M - 1) \geq 2 R_0.$$
  Now, note that $R_{\text{large}} \leq R_0 + 2R'$, and hence
  $R' \geq \frac{2 R_{\text{large}}}{5} \geq \frac{r(E)}{5}$,
  where the second inequality uses
  \Cref{lem:mostly_large_expansions}. Now, since each defensive ball
  has a significant core, the sum of radii of the cores of these $M$
  defensive balls is at least $\frac{\eps}{1 + \eps} \frac{r(E)}{5} \geq \frac{\eps}{6} r(E)$.
    \item $M'' \le M$. Here, $\min(k_{in}, M) \ge M''$. Again,
  each of these $M''$ large expansions is done by an amount $2r_{\max}$, where
  $r_{\max}$ is the radius of the largest defensive boundary ball at
  that point in time, and thus the sum of the corresponding $M''$
  largest boundary balls must be at least
  $\frac{R_{\text{large}}}{2} \geq \frac{r(E)}{4}$. The sum of
  the radii of their core optimal balls is therefore at least
  $\frac{\eps}{6} r(E)$. Since there are exactly $k_{in}$ optimal balls defended by defensive balls indexed in $J_{in}$, we can add any remaining $\min(k_{in}, M) - M''$ of these optimal balls to this set and the statement of the lemma follows. \qedhere
  \end{enumerate}
\end{proof}

\subsubsection{Phase 3: Identifying and Solving the Subproblems}
\label{sec:ident-subproblems}

Our algorithm sets $\cB = \{D_j : j \in J_b\}$ as the set of balls to be included in our solution, and then expands them (line \ref{l:expanded_boundary_ball}) as in~\Cref{lem:expansion_for_central_property} in order to satisfy the central property. Then, it sets $S_{in}$ as the set of all the points inside $E$ that were not covered by the expanded boundary balls (line \ref{l:s_in}), and $S_{out}$ as the set of points outside $E$ that
were not covered by the expanded boundary balls (line \ref{l:s_out}). ~\Cref{fig:msr_recursion_example} gives an example of the expansion of the boundary balls to obtain these sets of these points. We say that the pair $(S_{in}, k_{in})$ forms the \emph{inside subproblem}, and the pair $(S_{out}, k_{out})$ forms the \emph{outside subproblem}. The following lemma states that $\OPT(S_{in})$ and $\OPT(S_{out})$ do not interact with each other and that $k_{in}$ and $k_{out}$ are upper bounds on their sizes respectively. Thus, we can solve the inside and outside subproblems independently.

\begin{lemma}
  \label{lem:S_in_out_properties}
  Suppose that the clustering $(\{D_1, D_2, \ldots, D_{t'}\}, (k_1, k_2, \ldots, k_{t'}))$ chosen in line~\ref{l:clustering_loop} is a defensive clustering of $\OPT(S)$. Then, every optimal ball that contains a point from $S_{in}$ (respectively, $S_{out}$) must be defended by some ball indexed in $J_{in}$ (resp., $J_{out}$) and therefore:
  \begin{itemize}
    \item $\OPT(S_{in}) \cap \OPT(S_{out}) = \emptyset$.
    \item $|\OPT(S_{in})| \le k_{in}$ and $|\OPT(S_{out})| \le k_{out}$.
    \item $\opt(S_{in}) \le r(E).$
  \end{itemize}
\end{lemma}

\begin{proof}
  Consider any optimal ball $B_i^\star \in \OPT(S_{in})$. Let $D_j$ be the defensive ball that defends $B_i^\star$. By definition of $S_{in}$, the defensive balls on the boundary of $E$ do not cover any of the points inside $S_{in}$. Hence, $D_j$ must not be a boundary defensive ball. On the other hand, $B_i^\star$ clearly intersects $E$, and $B_i^\star \subseteq D_j$, so $D_j$ intersects $E$. Thus, $D_j$ must be fully inside $E$, i.e., $j \in J_{in}$. By similar arguments, every optimal ball containing a point in $S_{out}$ must be defended by a defensive ball fully outside $E$. Thus, $\OPT(S_{in}) \cap \OPT(S_{out}) = \emptyset$, $|\OPT(S_{in})| \leq k_{in}$, and $|\OPT(S_{out})| \leq k_{out}$.

Finally, since the optimal balls $\OPT(S_{in})$ are completely inside
the balls from $J_{in}$, their sum of radii must be upper bounded by $r(E)$. If not, then we can replace the optimal balls in $\OPT(S_{in})$ by $E$ to get a cheaper solution, contradicting the optimality of $\OPT(S)$.
\end{proof}

\begin{figure}[t]
  \centering
  \scalebox{0.6}{
\begin{tikzpicture}
    \usetikzlibrary{decorations.pathmorphing}
    \coordinate (O) at (0,0);
    \draw[dashed] (O) circle (4cm); %
     \draw[thick,black] (3,2.7) circle (0.95cm);
    \draw[thick,black] (-3.,-3) circle (0.85cm);
    \draw[thick,black] (-3.25,2) circle (1.75cm);
    \draw[thick,black] (3,-3) circle (1.3cm);
    \draw[thick,black] (3.3,1.4) circle (0.95cm);

    \fill[green!50!black] (3,2.7) circle (2pt);
    \fill[green!50!black] (-3,-3) circle (2pt);
    \fill[green!50!black] (-3.25,2) circle (2pt);
    \fill[green!50!black] (-4,2.5) circle (2pt);
    \fill[green!50!black] (3,-3) circle (2pt);
    \fill[green!50!black] (3.5,-2.5) circle (2pt);
    \fill[green!50!black] (3.3,2.1) circle (2pt);
    \fill[green!50!black] (3.1,2.4) circle (2pt);
    \fill[green!50!black] (2.5,2.9) circle (2pt);
    \fill[green!50!black] (3.6,2.6) circle (2pt);
    \fill[green!50!black] (-3.2,-2.7) circle (2pt);
    \fill[green!50!black] (-2.8,-3.4) circle (2pt);
    \fill[green!50!black] (-2.5,1.5) circle (2pt);
    \fill[green!50!black] (-3.8,1.2) circle (2pt);
    \fill[green!50!black] (-2.9,2.8) circle (2pt);
    \fill[green!50!black] (2.5,-2.8) circle (2pt);
    \fill[green!50!black] (3.4,-3.8) circle (2pt);
    \fill[green!50!black] (2.8,-3.5) circle (2pt);
    \fill[green!50!black] (3.3,1.1) circle (2pt);
    \fill[green!50!black] (3,3.5) circle (2pt);

    \fill[blue] (O) circle (2pt);
    \fill[blue] (1,1) circle (2pt);
    \fill[blue] (2.1,1.1) circle (2pt);
    \fill[blue] (-1,1.7) circle (2pt);
    \fill[blue] (O) circle (2pt);
    \fill[blue] (1,1) circle (2pt);
    \fill[blue] (2,-1.7) circle (2pt);
    \fill[blue] (1.7,2.2) circle (2pt);
    \fill[blue] (-2.7,-2) circle (2pt);
    \fill[blue] (-0.3,-2) circle (2pt);

    \fill[red] (3.7,3.7) circle (2pt);
    \fill[red] (-4.5,4.2) circle (2pt);
    \fill[red] (-4.8,-4.5) circle (2pt);
    \fill[red] (4.5,-4.1) circle (2pt);
    \fill[red] (0,4.8) circle (2pt);
    \fill[red] (-5.0,0) circle (2pt);
    \fill[red] (0,-4.8) circle (2pt);
    \fill[red] (4.8,0.5) circle (2pt);
    \fill[red] (4.3,4.3) circle (2pt);
    \fill[red] (3.9,4.1) circle (2pt);
    \fill[red] (-4.2,3.8) circle (2pt);
    \fill[red] (-3.8,-4.2) circle (2pt);
    \fill[red] (4.5,-2.5) circle (2pt);
    \fill[red] (-5.2,2.1) circle (2pt);

    \usetikzlibrary{shapes.geometric}
    \node[star, star points=5, star point ratio=2, fill=green!50!black, inner sep=1.2pt] at (2.5,1.5) {};
    \draw[thick, orange] (2.5, 1.5) circle (0.78cm);

    \node[star, star points=5, star point ratio=2, fill=blue, inner sep=1.2pt] at (1.9,1.9) {};
    \draw[thick, orange] (1.9, 1.9) circle (0.45cm);

    \node[star, star points=5, star point ratio=2, fill=green!50!black, inner sep=1.2pt] at (2.3,-2.1) {};
    \draw[thick, orange] (2.3,-2.1) circle (0.65cm);

    \node[star, star points=5, star point ratio=2, fill=green!50!black, inner sep=1.2pt] at (-4,3.4) {};
    \draw[thick, orange] (-4,3.4) circle (0.6cm);
    \draw[->, thick, double] (6.5,0) -- (8.5,0);

        \begin{scope}[shift={(15,0)}]
        \coordinate (O) at (0,0);

        \draw[dashed] (O) circle (4cm); %
        \draw[thick,black] (3,2.7) circle (0.95cm);
        \draw[thick,black] (-3.,-3) circle (0.85cm);
        \draw[thick,black] (-3.25,2) circle (2.32cm);
        \draw[thick,black] (3,-3) circle (1.94cm);
        \draw[thick,black] (3.3,1.4) circle (1.95cm);

        \fill[green!50!black] (3,2.7) circle (2pt);
        \fill[green!50!black] (-3,-3) circle (2pt);
        \fill[green!50!black] (-3.25,2) circle (2pt);
        \fill[green!50!black] (-4,2.5) circle (2pt);
        \fill[green!50!black] (3,-3) circle (2pt);
        \fill[green!50!black] (3.5,-2.5) circle (2pt);
        \fill[green!50!black] (3.3,2.1) circle (2pt);
        \fill[green!50!black] (3.1,2.4) circle (2pt);
        \fill[green!50!black] (2.5,2.9) circle (2pt);
        \fill[green!50!black] (3.6,2.6) circle (2pt);
        \fill[green!50!black] (-3.2,-2.7) circle (2pt);
        \fill[green!50!black] (-2.8,-3.4) circle (2pt);
        \fill[green!50!black] (-2.5,1.5) circle (2pt);
        \fill[green!50!black] (-3.8,1.2) circle (2pt);
        \fill[green!50!black] (-2.9,2.8) circle (2pt);
        \fill[green!50!black] (2.5,-2.8) circle (2pt);
        \fill[green!50!black] (3.4,-3.8) circle (2pt);
        \fill[green!50!black] (2.8,-3.5) circle (2pt);
        \fill[green!50!black] (3.3,1.1) circle (2pt);
        \fill[green!50!black] (3,3.5) circle (2pt);

        \fill[blue] (O) circle (2pt);
        \fill[blue] (1,1) circle (2pt);
        \fill[green!50!black] (2.1,1.1) circle (2pt);
        \fill[green!50!black] (-1,1.7) circle (2pt);
        \fill[blue] (O) circle (2pt);
        \fill[blue] (1,1) circle (2pt);
        \fill[green!50!black] (2,-1.7) circle (2pt);
        \fill[green!50!black] (1.7,2.2) circle (2pt);
        \fill[blue] (-2.7,-2) circle (2pt);
        \fill[blue] (-0.3,-2) circle (2pt);

        \fill[red] (3.7,3.7) circle (2pt);
        \fill[red] (-4.5,4.2) circle (2pt);
        \fill[red] (-4.8,-4.5) circle (2pt);
        \fill[green!50!black] (4.5,-4.1) circle (2pt);
        \fill[red] (0,4.8) circle (2pt);
        \fill[red] (-5.0,0) circle (2pt);
        \fill[red] (0,-4.8) circle (2pt);
        \fill[green!50!black] (4.8,0.5) circle (2pt);
        \fill[red] (4.3,4.3) circle (2pt);
        \fill[red] (3.9,4.1) circle (2pt);
        \fill[green!50!black] (-4.2,3.8) circle (2pt);
        \fill[red] (-3.8,-4.2) circle (2pt);
        \fill[green!50!black] (4.5,-2.5) circle (2pt);
        \fill[green!50!black] (-5.2,2.1) circle (2pt);

        \node[star, star points=5, star point ratio=2, fill=green!50!black, inner sep=1.2pt] at (2.5,1.5) {};
        \draw[thick, orange] (2.5, 1.5) circle (0.78cm);

        \node[star, star points=5, star point ratio=2, fill=green!50!black, inner sep=1.2pt] at (1.9,1.9) {};
        \draw[thick, orange] (1.9, 1.9) circle (0.45cm);

        \node[star, star points=5, star point ratio=2, fill=green!50!black, inner sep=1.2pt] at (2.3,-2.1) {};
        \draw[thick, orange] (2.3,-2.1) circle (0.65cm);

        \node[star, star points=5, star point ratio=2, fill=green!50!black, inner sep=1.2pt] at (-4,3.4) {};
        \draw[thick, orange] (-4,3.4) circle (0.6cm);
    \end{scope}
\end{tikzpicture}
}
\caption{An example of the creation of subproblems. The left image describes the boundary balls before we expand them in order to satisfy the central property, and the image on the right shows their status after this expansion. The dashed circle represents the expanded ball $E$ and the solid circles denote the boundary balls. The orange balls represent some optimal balls whose centers were (either initially or during the process of expanding the boundary balls) inside the union of the boundary balls, but they were not completely covered. The blue points denote $S_{in}$, the red points denote $S_{out}$ and the green points denote the points covered by the (expanded) boundary balls.}
  \label{fig:msr_recursion_example}
\end{figure}
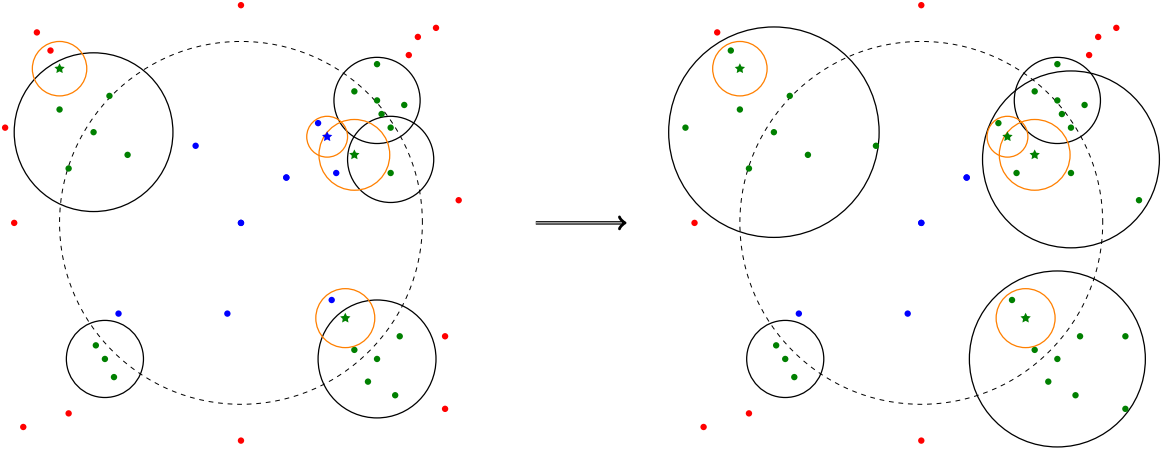

For solving the outer subproblem, we simply call
$\recMSR(S_{out}, k_{out}, \ell)$ without increasing the level. For
solving the inside subproblem, we call
${\tt guessAndSolveMSR}(S_{in}, k_{in}, \ell + 1)$---which is presented in
\Cref{alg:guess_rec_msr}---where we ``guess'' (by iterating over all
$n^{O(M)}$ possibilities) the $M' = \min(k_{in}, M)$ biggest optimal
balls (say $\{B''_1, B''_2, \ldots, B''_{M'}\}$) of
$\OPT(S_{in})$. Then, we include all these $M'$ balls in our solution
and then make a recursive call
$\recMSR(S_{in} \setminus \cup_{j \in [M']} B''_j, k_{in} - M', \ell +
1)$ at level $\ell + 1$ to cover the points of $S_{in}$ not covered by
these $m$ balls, using at most $k_{in} - M'$ clusters. Note that we do
not need to expand the guessed balls to satisfy the central property
since these are optimal balls and hence they already satisfy the
central property (no optimal ball contains the center of another). We
return the best solution obtained out of all the guesses. Note that the total number of balls used by the solution is at most:
\begin{align*}
    k_{in} + k_{out}+|J_b| &= \sum_{j \in J_{in}} k_j + \sum_{j \in J_{out}} k_j + \sum_{j \in J_b} 1\\
    &\le \sum_{j \in J_{in}} k_j + \sum_{j \in J_{out}} k_j + \sum_{j \in J_b} k_j\\
    &\le t.
\end{align*}
where the inequalities use the fact the tuple $(k_1, k_2, \ldots k_{t'})$ consists of positive integers with sum not exceeding $t$ and that $J_{in}, J_{out}, J_b$ form a partition of $[t']$. Also, recall from~\Cref{lem:constant_balls_cover_constant_fraction} that up
to $M$ optimal balls within $E$ account for at least $\frac{\eps}{6}$
fraction of $r(E)$. The following lemma states that the same can
be said about the union of the optimal balls covered by the expanded
boundary balls and the guessed $M'$ largest optimal balls of
$\OPT(S_{in})$. Generally speaking, this is because if any of the $M$
largest optimal balls within $E$ is not fully covered by the expanded
boundary balls, it must be one of the $M$ largest balls in
$\OPT(S_{in})$ and hence gets guessed in the next step. This
observation helps us make progress at each level of the recursion
tree, since we only pay a factor of $1$ for these optimal balls,
except for the cores of the boundary balls. For the cores of the
boundary balls, we pay a factor of
$2$(see~\Cref{lem:expansion_for_central_property}) but they only have
a total cost of $O(\eps^2 r(E))$.

\begin{lemma}
  \label{lem:constant_fraction_saved}
  Let $T$ be the set of indices of the optimal balls of $\OPT(S)$
  completely covered by the set $\cB'$ of the expanded defensive
  boundary balls, and let $M' = \min(k_{in}, M)$. Let $I'$ be the set of the indices(in $\OPT(S)$)
  of the largest $\min(M', |\OPT(S_{in})|)$ optimal balls of
  $\OPT(S_{in})$. Then, $$\sum_{i \in T \cup I'} r_i^\star \geq \frac{\eps}{6}
  r(E).$$
\end{lemma}
\begin{proof}
    Breaking ties globally by their index in $\OPT(S)$, consider the $M'$ largest (out of the $k_{in}$ total) optimal balls defended by defensive balls with an index in $J_{in}$. Any of these balls that does not get covered by $\cB'$ must be one of the largest $\min(M', |\OPT(S_{in})|)$ optimal balls of $\OPT(S_{in})$ according to the same global tie-breaking rule. Thus, its index must be in $I'$, and the statement of the lemma follows from~\Cref{lem:constant_balls_cover_constant_fraction}.
\end{proof}

\begin{algorithm}[h]
  \caption{{\tt guessAndSolveMSR} procedure}
  \label{alg:guess_rec_msr}
  \DontPrintSemicolon
  \textbf{Global:} The metric space $(X, d)$, the parameters $k, M, L$\;
  \KwIn{A non empty subset $S\subseteq X$, $t \in [k]$, $\ell \in \{1,\ldots,L\}$}
  \KwOut{A clustering of $S$ with at most $t$ balls.}
  \lIf{$|S| = 0$}{
    \Return $\emptyset$
  }
  $\cA^\star \gets null$ \tcp*{assume $\cost(null) = \infty$}
  balls $\gets \{\Ball(p, r): p \in S, r \in \{d(p, q): q \in X\}\}$ \tcp*{set of all possible balls}
  \For{tuples $(B''_1, B''_2, \ldots, B''_{M'})$ of balls with $M' = \min(t, M)$\label{l:guessing_loop}}{
    $\cA \gets \{B''_1, \ldots, B''_{M'}\} \cup \recMSR(S \setminus \cup_{j \in [M']} B''_j, t - M', \ell)$\;
    \If{the cost of $\cA$ is less than the cost of $\cA^\star$}{
      $\cA^\star \gets \cA$\;
    }
  }
  \Return $\cA^\star$\;
\end{algorithm}

\subsection{Analysis of the approximation factor}
\label{sec:analysis-MSR}
Recall from~\Cref{sec:general_approach_msr} that we call a pair $(S, t)$ \emph{valid} if $|\OPT(S)| \le t$ and the center of each optimal ball in $\OPT(S)$ is present in $S$. As mentioned in~\Cref{sec:general_approach_msr}, our analysis will rely on valid pairs, even though these are not the only kind of inputs passed to \recMSR.

Note that the parameters $(X, k)$ of the initial call to
$\recMSR(X, k, 0)$ form a valid pair since $|\OPT(X)| \le k$ and the
centers of all the optimal balls are present in $X$. We assume that
$\eps$ is a small enough positive constant. Recall that $\opt(S)$ denotes the sum of radii of all the optimal balls in $\OPT(S)$ and 
$L = \ceil{\frac{6}{\eps} \ln \frac{1}{\eps}}$. We will prove the
following upper bound on the approximation factor:
\begin{lemma}
  \label{lem:approx_ratio}
  For each $\ell \in \{0, 1, \ldots, L\}$ and each valid pair
  $(S, t)$, the cost of the solution returned by $\recMSR(S, t, \ell)$
  is at most $\rho(\ell) \cdot \opt(S)$, where
  $\rho(\ell) = 1 + 25\eps + (1 - \frac{\eps}{6})^{L - \ell}$.
\end{lemma}

Intuitively, the lemma implies that we get increasingly precise
solutions as we decrease the level from $L$ to $0$. Substituting
$\ell = 0, S = X$ and $t = k$, we get that $\recMSR(X, k, 0)$ returns
a $(1 + 26\eps)$-approximation for the original problem when
$\ell = 0$.

\begin{proof}[Proof of \Cref{lem:approx_ratio}.]
  We will use induction on $|S| + L - \ell$. When $\ell = L$, we have
  $\rho(L) \geq 2$, and at level $L$ we simply return the best
  solution out of the candidate defensive clusterings $\cC$. At least
  one of these is a defensive clustering, and each defensive
  clustering is a $2$-approximation. Also, when $|S| = 0$, our
  algorithm returns an empty clustering with cost $0$.
  
  When $S$ is non-empty and $\ell < L$, our solution first guesses a
  defensive clustering of $\OPT(S)$ (by iterating over a set of
  candidate clusterings guaranteed to contain one since $(S, t)$ is
  valid), say $D_1, D_2, \ldots, D_{t'}$. Each defensive ball $D_j$
  defends a set of optimal balls with indices in $I(j)$, and has a
  core with index $\kappa(j)$.

  There are two cases.  
  The first case is where there exists a defensive ball $D_j$ with an
  insignificant core
  $(r_{\kappa(j)}^\star \le \eps\sum_{i \in I(j)}r_i^\star)$. Here, in
  line~\ref{l:good_ball_expansion}, we expand $D_j$ to get a ball $D'_j$
  in order to satisfy the central property. Then, we include $D'_j$ in
  our solution and recursively find a clustering for the set
  $S \setminus D'_j$ with an upper bound of $t - 1$ on the number of clusters. Clearly, $(S \setminus D'_j, t-1)$ forms a valid pair since $|\OPT(S)| \le t$ and the core optimal ball of $D_j$ is present in $\OPT(S)$ but not in $\OPT(S \setminus D'_j)$. Thus, setting $\cB' = \{D'_j\}$
  in~\Cref{lem:expansion_for_central_property} and applying the
  induction hypothesis, we get that the total cost is therefore at most:
  $$\cost(\cB') + \rho(\ell) \opt(S \setminus D'_j) \leq r_{\kappa(j)}^\star + \sum_{i \in T} r_i^\star + \rho(\ell) \opt(S \setminus D'_j).$$
  where $T$ is the set of balls from $\OPT(S)$ completely inside $D'_j$. Clearly, $\opt(S \setminus D'_j) \leq \opt(S) - \sum_{i \in T} r_i^\star$. Also, since $D_j$ has an insignificant core and $I(j) \subseteq T$, we have:
  $$r_{\kappa(j)}^\star \leq \eps \sum_{i \in I(j)}r_i^\star \leq \eps \sum_{i \in T} r_i^\star.$$
  Thus, the total cost is at
  most:$$\rho(\ell) \opt(S) + (1 + \eps - \rho(\ell)) \sum_{i \in T}
  r_i^\star \leq \rho(\ell) \opt(S),$$ using
  $\rho(\ell) \geq 1 + 25\eps $ for all $\ell \in [0, L]$.

  The second case is when each ball has a significant core
  $(r_{\kappa(j)}^\star > \eps\sum_{i \in I(j)}r_i^\star)$. In this
  case, we form an expanded ball $E$ and get a set of boundary balls with indices $J_b$ to include in our solution. We expand them in order to satisfy the central property to get a set of balls $\cB'$ to include in our solution. From~\Cref{lem:expansion_for_central_property,lem:small_boundary_core}, we get:
  \begin{equation}
    \label{eqn:cost_B'}
    \cost(\cB') \le 4\eps^2 r(E) + \sum_{i \in T} r_i^\star.
  \end{equation}
  where $T$ denotes the set of optimal balls of $\OPT(S)$ fully covered
  by $\cB'$. Then, we obtain subproblems $S_{in}$ and $S_{out}$ as the set of
  points not covered by $\cB'$ inside and outside $E$ respectively. Recall from~\Cref{lem:S_in_out_properties} that $\OPT(S_{in}) \cap \OPT(S_{out}) = \emptyset$, $|\OPT(S_{in})| \le k_{in}$, $|\OPT(S_{out})| \le k_{out}$ and $\opt(S_{in}) \le r(E)$. Additionally, since $\cB'$ satisfies the central property, the centers of the optimal balls in $\OPT(S_{in})$ and $\OPT(S_{out})$ must be in $S_{in}$ and $S_{out}$ respectively. Thus, $(S_{in}, k_{in})$ and $(S_{out}, k_{out})$ form valid pairs. After including $\cB'$ in our solution, we recursively compute the following two solutions:

  \begin{enumerate}
  \item $\cB_{out}^\star = \recMSR(S_{out}, k_{out}, \ell)$. By the induction hypothesis, 
    \begin{equation}
      \label{eqn:cost_B_out}
      \cost(\cB_{out}^\star) \le \rho(\ell) \cdot \opt(S_{out}).
    \end{equation}
  \item $\cB^\star_{in} = {\tt guessAndSolveMSR}(S_{in}, k_{in}, \ell + 1)$,
    which is computed by considering all possible tuples
    $\cB'' = (B''_1, \ldots, B''_{M'})$ of $M' = \min(k_{in}, M)$ balls
    and computing
    $\recMSR(S_{in} \setminus \cup_{i \in [M']} B''_i, k_{in} - M', \ell
    + 1)$. Let $I'$ be the set of indices (in $\OPT(S)$) of the $\min(M', |\OPT(S_{in})|)$
    largest optimal balls of $\OPT(S_{in})$. Consider a tuple $\cB''$ of $M'$ balls that consists of all the optimal balls having indices in $I'$, along with $M' - \min(M', |\OPT(S_{in})|)$ additional balls each having radius $0$ and centered at the center of one of the optimal balls in $I'$. Since $(S_{in}, k_{in})$ is
    a valid pair, so is
    $(S_{in} \setminus \cup_{i \in I'}B_i^\star, k_{in} - M')$ as
    optimal balls do not contain the centers of other optimal
    balls and $\OPT(S \setminus \cup_{i \in I'} B_i^\star) \subseteq \OPT(S) \setminus \{B_i^\star: i \in I'\}$. Thus, by the induction hypothesis:
    \begin{equation}
      \label{eqn:cost_B_in}
      \cost(\cB^\star_{in}) \leq \rho(\ell + 1) \cdot \opt(S_{in} \setminus \cup_{i \in I'} B_i^\star) + \sum_{i\in I'} r_i^\star.
    \end{equation}
  \end{enumerate}

  By~\Cref{eqn:cost_B',eqn:cost_B_out,eqn:cost_B_in}, we get that the
  cost of the combined solution, i.e.,
  $\cost(\cB') + \cost(\cB_{out}^\star) + \cost(\cB_{in}^\star)$, is at
  most:
  \begin{equation}
    \label{eqn:total_cost}
    \begin{aligned}
      &4\eps^2 r(E)+\rho(\ell) \opt(S_{out}) + \rho(\ell+1)\opt(S_{in} \setminus \cup_{i \in I'} B_i^\star) + \sum_{i \in T \cup I'} r_i^\star\\
      &\leq 4\eps^2 r(E)+\rho(\ell)\bigl(\opt(S_{out}) + \opt(S_{in} \setminus \cup_{i \in I'} B_i^\star)\bigr) + (\rho(\ell+1) - \rho(\ell)) \opt(S_{in}) + \sum_{i \in T \cup I'} r_i^\star\\
      &\leq 4\eps^2 r(E) + \rho(\ell)\Big(\opt(S_{in}) + \opt(S_{out}) - \sum_{i \in I'} r_i^\star\Big) + (\rho(\ell + 1) - \rho(\ell)) r(E) + \sum_{i\in T\cup I'}r_i^\star
    \end{aligned}
  \end{equation}
  where the first inequality follows from
  $\opt(S_{in} \setminus \cup_{i \in I'}B_i^\star) \leq
  \opt(S_{in})$. The second inequality follows from
  $\opt(S_{in} \setminus \cup_{i \in I'}B_i^\star) \leq \opt(S_{in}) -
  \sum_{i \in I'}r_i^\star$ and $\opt(S_{in}) \leq r(E)$. Now,
  since $\OPT(S_{in})$ and $\OPT(S_{out})$ are disjoint subsets of $\OPT(S)$ and the optimal
  balls with indices in $T$ do not appear in them (since they were
  already covered by $\cB'$ and hence the points within them were
  removed), we have:
  \begin{equation}
    \label{eqn:in_out_sum}
    \opt(S_{in}) + \opt(S_{out}) \leq \opt(S) - \sum_{i \in T} r_i^\star.
  \end{equation}
  Also, from~\Cref{lem:constant_fraction_saved}, we have:
  \begin{equation}
    \label{eqn:covered_sum}
    \sum_{i \in T \cup I'} r_i^\star \geq \frac{\eps}{6} r(E).
  \end{equation}

  Thus, substituting \Cref{eqn:in_out_sum,eqn:covered_sum} in~\Cref{eqn:total_cost}, we get that the cost of the output clustering is at most:
  \begin{align*} 
    &4\eps^2 r(E) + \rho(\ell)(\opt(S) - \sum_{i \in T \cup I'} r_i^\star) + (\rho(\ell + 1) - \rho(\ell)) r(E) + \sum_{i \in T \cup I'} r_i^\star\\
    &= \rho(\ell)\opt(S) + (\rho(\ell + 1) - \rho(\ell) + 4\eps^2) r(E) - (\rho(\ell) - 1) \sum_{i \in T \cup I'} r_i^\star\\
    &\leq \rho(\ell)\opt(S) + (\rho(\ell + 1) - \rho(\ell) + 4\eps^2) r(E) - (\rho(\ell) - 1) \frac{\eps}{6} r(E)\\
    &= \rho(\ell) \opt(S) - \left(\frac{\eps(\rho(\ell) - 1)}{6} - (\rho(\ell + 1) - \rho(\ell)) - 4\eps^2\right)r(E).
  \end{align*}
  Finally, substitute $\rho(\ell) = 1 + 25\eps + (1 - \frac{\eps}{6})^{L - \ell}$ to note that:
  \begin{equation}
    \label{eqn:rho_eqn}
    \begin{aligned}
      \frac{\eps (\rho(\ell) - 1)}{6} - (\rho(\ell + 1) - \rho(\ell)) -
      4\eps^2
      &= \frac{\eps}{6}\left(25\eps + (1 - \frac{\eps}{6})^{L - \ell}\right) - \frac{\eps}{6}(1 - \frac{\eps}{6})^{L - \ell - 1}  - 4\eps^2\\
      &= \frac{25}{6}\eps^2 - \frac{\eps}{6}\left(1 - \frac{\eps}{6}\right)^{L - \ell - 1} \frac{\eps}{6} - 4\eps^2\\
      &= \frac{\eps^2}{6} -\frac{\eps^2}{36} \left(1 - \frac{\eps}{6}\right)^{L - \ell - 1} \geq 0.
    \end{aligned}
  \end{equation}
  This completes the proof of \Cref{lem:approx_ratio}.
\end{proof}

\subsection{Time complexity of $\recMSR$}
\label{sec:time-MSR}

To complete the proof of \Cref{thr:mainMSR}, we now bound the runtime
of the algorithm. Consider the recursion tree of $\recMSR$ and let
$f(\ell, t)$ denote the maximum possible number of nodes in the
subtree rooted at a node corresponding to a call
$\recMSR(S, t, \ell)$.

All the for loops in $\recMSR$, have at most $(\frac{1}{\eps})^{O(k)}$
iterations since $|\cC| \leq (\frac{1}{\eps})^{O(k)}$, and the number
of $t$-tuples of non negative integers with sum $O(\frac{k}{\eps})$,
for any $t \le k$ is also $(\frac{1}{\eps})^{O(k)}$. Also, the for
loop in ${\tt guessAndSolveMSR}$ has at most $n^{2M}$
iterations. Thus, there exists $b = (\frac{1}{\eps})^{O(k)}$, such
that $\recMSR(S, t, \ell)$ makes at most $b$ calls to subproblems at
level $\ell$, and at most $b \cdot n^{2M}$ calls to subproblems at level
$\ell + 1$. Hence, we have that for all $t \geq 1, \ell \in \{0, 1, \ldots, L-1\}$:
$$f(\ell, t) \leq b f(\ell, t - 1) + b n^{2M}f(\ell + 1, t - 1) + 1.$$
The following lemma proves an FPT upper bound on the size of the recursion tree.
\begin{lemma}
  $f(\ell, t) \leq b^{2(t+L - \ell)} n^{2M (L - \ell)}$ for all $t \in \{0, 1, \ldots, k\}, \ell \in \{0, 1, \ldots, L\}$
\end{lemma}
\begin{proof}
  We prove this via induction on $L - \ell + t$. When $\ell = L$ or $t = 0$, there are no recursive calls and hence the statement is trivially true. We have:
  \begin{align*}
    f(\ell, t) &\leq bf(\ell, t - 1) + b n^{2M} f(\ell + 1, t - 1) + 1\\
               &\leq b \cdot b^{2(t + L - \ell - 1)} n^{2M (L - \ell)} +  bn^{2M} \cdot b^{2(t + L - \ell-2)} n^{2M (L - \ell - 1)} + 1\\
               &\leq 2 b^{2(t + L - \ell) - 1} n^{2M (L - \ell)} + 1\\
               &\leq b^{2(t + L - \ell)} n^{2M(L - \ell)}. \qedhere
  \end{align*}
\end{proof}
In each node of the recursion tree, we spend $(\frac{1}{\eps})^{O(k)} \cdot n^{2M}$ time. Thus, substituting $b = (\frac{1}{\eps})^{O(k)}$, $L = \ceil{\frac{6}{\eps} \ln \frac{1}{\eps}}$ and $M = \ceil{\frac{1}{\eps^2}}$, we get the following result:

\begin{corollary}
  \label{lem:time_complexity_msr}
  The overall time complexity of $\recMSR$ is
  $$\left(\frac{1}{\eps}\right)^{O\left(k^2 + \nicefrac{k}{\eps} \log \nicefrac{1}{\eps}\right)} n^{O\left(\nicefrac{1}{\eps^3} \log \nicefrac{1}{\eps}\right)}.$$
\end{corollary}

\Cref{lem:approx_ratio} and \Cref{lem:time_complexity_msr} complete
our discussion of the first $(1+\eps)$-approximation for \MSR in FPT
time $f(k,\eps) n^{\poly(1/\eps)}$.

The claimed runtime in \Cref{lem:time_complexity_msr}, though FPT, is worse than
that claimed in \Cref{thr:mainMSR}: it contains a $k^2$ term in the
exponent in the runtime. To see how to get the
improvement, note that the $k^2$ term appears in the exponent because
there are $(1/\eps)^{O(k)}$ candidate defensive clusterings; for each
of them, we make a recursive call at the same level $\ell$ (in
line~\ref{l:insignificant_core_recursion} and
line~\ref{l:outside_recursion}) potentially reducing $t$ only by
$1$. In the modified algorithm, we will instead make all recursive
calls at the next level, $\ell + 1$. We begin by guessing the set of
``good'' defensive balls (whose core is insignificant). Then, after we
find the first expanded ball $E$ and the set of boundary balls,
instead of solving the outside subproblem at the same level, we will
keep probing the outside points (if any) to form another expanded
ball. We repeat this until all defensive balls are either good, or
completely inside/on the boundary of one of the expanded balls. We
defer the details to~\Cref{sec:faster_msr}.

\section{Concluding Remarks}
\label{sec:concluding-remarks}

We consider the min-sum diameters and min-sum radii problems, where we
partition a set of points into $k$ parts, to minimize the sum of the
diameters or radii of the parts.  For both problems, we present the
first FPT approximation schemes, with respect to the natural parameter
$k$; the previous best FPT approximation algorithms for these problems
have approximation factors $4+\eps$ and $2+\eps$,
respectively. Specifically, given $\eps>0$, we show how to compute
$(1+\eps)$-approximations for both \MSD and \MSR in time
$(1/\eps)^kn^{O(1)}$ and
$(1/\eps)^{O(k/\eps \log 1/\eps)}n^{\poly(1/\eps)}$ respectively.

Our work suggests several open problems. The question of whether we can get an \emph{exact} algorithm for \MSD in time parameterized by $k$ remains a tantalizing open problem: no $W[1]$-hardness is known for this problem. Improving the approximation ratio of the \MSR problem under mergeable constraints to get an FPT PTAS remains open; we only give a $(2+\eps)$-approximation. It would also be interesting to improve the runtimes of our algorithms further, and remove the dependence on $\eps$ in the exponent of $n$ for the \MSR problem.

\appendix
\crefalias{section}{appendix}
\section{Missing Proofs}
\label{sec:missing}

\begin{proof}[Proof of \Cref{lem:reduction:radii}]
Consider an input instance $(X,d,k)$ of \MSR with optimal cost $\opt$. %
We will assume without loss of
  generality, that $\opt>0$, else the set $X$ consists of at
  most $k$ subsets having zero radius, and hence the optimal solution
  can be computed in polynomial time. Now, in deterministic polynomial time, compute
  an approximate solution $\APX$ having cost $\apx$ lying in the
  interval $[\opt, 4 \opt]$, e.g., using the algorithm in
  \cite{doi:10.1137/1.9781611977912.69}.

  Given this estimate of the optimal value, modify the metric by
  rounding up each pairwise distance to the next integer multiple of
  $\nicefrac{\eps}{4 k}\cdot\apx$, and recomputing the metric
  closure. Since the radius of each cluster in the original optimal
  solution increases by at most $\nicefrac{\eps}{4 k}\cdot\apx$, this
  rounding may increase the cost of the optimum solution at most by
  $\eps \opt$, hence the optimum solution is at most
  $(1+\eps)\,\opt\leq (1+\eps)\,\apx$. Moreover, all the distances are
  now integer multiples of $\nicefrac{\eps}{4 k}\cdot\apx$, so
  multiplying these distances by $\nicefrac{4 k}{(\eps \apx)}$ gives
  us a metric space $\cM' = (X,d')$ with integer distances.

  The optimum solution value for $\cM'$ is at most
  $\frac{4k}{\eps \apx}\cdot (1+\eps)\apx\leq \frac{8k}{\eps}$. To
  this instance we apply the $\rho(\eps)$-approximation in the claim,
  and we return the obtained solution for the original problem. By
  construction, the cost of the obtained solution is at most
  $(1+\eps)\,\rho(\eps) \cdot \opt$. The overall running time is
  $T(n,k,\eps)+n^{O(1)}$, with the additive term due to the above
  preprocessing steps.

  The same preprocessing algorithm also works for $\MSD$ with slightly different constants, since we can get $\apx \in [\opt, 8\opt]$ (by running the algorithm from~\cite{doi:10.1137/1.9781611977912.69}) where $\opt$ now denotes the cost of the optimal \MSD clustering.
\end{proof}

\section{Derandomization of the MSD Algorithm}
\label{sec:derandomization}

In this section, we give a derandomization of the previous algorithm. First, we preprocess the input in $n^{O(1)}$ (using \Cref{lem:reduction:radii}) to get a new instance with integral distances and optimum bounded by $\frac{16k}{\eps}$, at a cost of $(1 + O(\eps))$ to the approximation factor. We will also update each distance $> \frac{16k}{\eps}$ to $\frac{16k}{\eps}$. Doing this preserves the optimum solution. In the rest of the section, we will provide an exact algorithm for such instances implying a $(1 + O(\eps))$-approximate algorithm for general instances with an additional additive overhead of $n^{O(1)}$ in the runtime.

The idea of our algorithm \detRecMSD(\Cref{alg:PTAS_MSD_det}) is
simple: Instead of randomly choosing the threshold in $[0, \diam(S)]$,
we will try each possible integral threshold $R \in \{0, 1, \ldots
\diam(S) - 1\}$ (recall that distances are now integers). The
following two lemmas establish the correctness and runtime of our
algorithm, thus proving~\Cref{thr:mainMSD}.

\begin{algorithm}[h]
  \caption{The Procedure \detRecMSD}
  \label{alg:PTAS_MSD_det}
  \DontPrintSemicolon
  \textbf{Global:} The metric space $(X, d)$ with integral distances in $[0, \frac{16k}{\eps}].$\;
  \KwIn{subset $S\subseteq X$ with $S\neq \emptyset$, and parameter $t \in [k].$}
  \KwOut{A vector $\cC = (\cC_1, \cC_2, \ldots, \cC_t)$, where $\cC_i$ is an $i$-clustering of $S$.}
  $\cC_i \gets \{S\}$ for all $i \in \{1, 2, \ldots, t\}$\;
  \lIf{$t=1$ or $|S| = 1$}{\textbf{return} $(\cC_1, \cC_2, \ldots
    \cC_t)$}
  Let $x, y \in S$ be such that $d(x, y) = \diam(S)$\;
  \For{$R = 0, 1, \ldots \diam(S) - 1$\label{l:cut_and_recurse}}{
        Let $S_1\leftarrow S\cap \Ball(x,R)$ and $S_2\leftarrow S\setminus \Ball(x,R)$\;
        Let $\cA \gets \detRecMSD(S_1, t - 1)$ and $\cB \gets \detRecMSD(S_2, t-1)$\;
      \tcp{Perform a $(\min, +)$ convolution of $\cA$ and $\cB$}
        \For{i, j = $1, 2, \ldots t - 1$}{
            \lIf{$i+j \leq t$ and cost$(\cA_i \cup \cB_j) <$ cost$(\cC_{i+j})$}{
            $\cC_{i+j} \gets \cA_i \cup \cB_j$
            }
      }
  }
  \Return $(\cC_1, \cC_2, \ldots, \cC_t)$
\end{algorithm}
\begin{lemma}
    \label{lem:msd_det_correctness}
    Let $\opt_i(S)$ denote the cost of the optimal $i$-clustering of
    $S$ and $(\cC_1, \cC_2, \ldots \cC_t)$ be the output of
    $\detRecMSD(S, t)$. Then, for each $r \in \{1, 2, \ldots t\}$, $\cost(\cC_r) = \opt_r(S)$.
\end{lemma}

\begin{proof}
    We prove this lemma via induction on $t$. If $t = 1$, there is only one $t$-clustering, which is also returned by the algorithm. Similar to the proof of~\Cref{lem:msd_prob}, consider the clusters $X_1, X_2, \ldots, X_{r'}$ in the optimal $r$-clustering of $S$, where $r' \leq r$. For each cluster $X_h$, define $L_h$ to be the smallest distance from $x$ to some point in $X_h$ and $R_h$ to be the largest distance from $x$ to some point in $X_h$, and define $I_h$ to be the interval $[L_h, R_h]$. Observe that the length of each interval $I_h$ is at most $\diam(X_h)$ and thus the sum of lengths of the $r'$ intervals corresponding to the optimal clusters is at most $\opt_r(S)$.

    If $\opt_r(S) \geq \diam(S)$, then the single cluster $\{S\}$ is already optimal. Assume $\opt_r(S) < \diam(S)$. Since all distances are integers, the sum of the lengths of the intervals $I_h$ is strictly less than the length of the interval $[0, \diam(S)]$. Thus, there exists some $R$, such that every interval is disjoint from $(R, R+1)$. Thus for each cluster, either all points are at a distance $\leq R$, or all points are at a distance $> R$. When this $R$ is chosen as the threshold in line~\ref{l:cut_and_recurse} of the algorithm, we get a split $S_1, S_2$ where $S_1$ is precisely a union of some $r_1$ clusters of the points at distance $\leq R$ from $x$ and $S_2$ is the union of some $r_2' = r'-r_1$ clusters of points at distance $>R$ from $x$. Let $r_2 = r - r_1 \geq r_2'$. The algorithm then computes $\cA$ as the output of $\detRecMSD(S_1, t-1)$ and $\cB$ as the output of $\detRecMSD(S_2, t-1)$. By the induction hypothesis, we have:
    \begin{gather}
      \cost(\cA_{r_1}) = \opt_{r_1}(S_1) \quad\text{and}\quad
      \cost(\cB_{r_2}) = \opt_{r_2}(S_2) \leq \opt_{r_2'}(S_2)
    \end{gather}

    Hence, after we compute the $(\min,+)$ convolution of $\cA$ and $\cB$, we must have:
    $$\text{cost}(\cC_r) \leq \text{cost}(\cA_{r_1}) +
      \text{cost}(\cB_{r_2}) \leq \opt_{r_1}(S_1) +
        \opt_{r_2'}(S_2) \leq \opt_r(S).$$
\end{proof}

\begin{lemma}
    $\detRecMSD(X, k)$ runs in time $\left(O(\frac{k}{\eps})\right)^k \cdot n^{O(1)}$
\end{lemma}
\begin{proof}
    For each problem in the recursion tree, we create $O(k/\eps)$ smaller subproblems corresponding to the $\diam(S) \leq \frac{16k}{\eps}$ possible values of the threshold $R$. Hence, the recursion tree has a branching factor of at most $O(k/\eps)$, and a depth of $k$. Thus, there are $(O(\frac{k}{\eps}))^k$ nodes in the recursion tree, and we spend $n^{O(1)}$ time in each node and the lemma statement follows.
\end{proof}

\section{Simple and Fast Exact Algorithm for \MSD}
\label{sec:exact_msd}

In this section, we describe a simple exact algorithm for \MSD that
runs in $n^{k+O(1)}$, improving the $n^{5k + O(1)}$ bound
from~\cite{10.1609/aaai.v39i15.33699}.

To solve the problem on some subset $S \sse X$, again consider a
diameter witness pair $x,y$; i.e., such that $d(x,y) =
\diam(S)$. Order the points of $S$ on a line as
$z_1 = x, z_2, \ldots z_{|S|} = y$ in increasing order of distance
from $x$, breaking ties arbitrarily. As a thought experiment, project
each cluster $X_h$ in the optimal solution onto an interval
$I_h = [L_h, R_h]$ where $L_h := \min \{i \in [|S|]: z_i \in X_h\}$
and $R_h := \max \{i \in [|S|]: z_i \in X_h\}$. If there exists some
$q \in [|S| - 1]$ such that $q$ and $q + 1$ are never a part of the same
interval, then we can split the problem into subproblems
$\{z_1, \ldots z_q\}$ and $\{z_{q+1}, \ldots z_{|S|}\}$. On the other
hand, when no such $q$ exists, it is optimal to have all the points in
the single cluster. This is because, the cost of the optimal
clustering equals the sum of
$\diam(X_h) \geq d(x, z_{R_h}) - d(x, z_{L_h})$ over all indices $h$
of the optimal clusters, and each $q \in [|S|-1]$ contributes at least
$d(x, z_{q+1}) - d(x, z_{q})$ to at least one term of this sum. Hence, the optimal cost is at
least $d(x, z_{|S|}) - d(x, z_1) = \diam(S)$. Since our algorithm does not
know the projected intervals, it tries all these options and chooses
the best one.

To bound the runtime, observe that the number of nodes in the
recursion tree is given by $T(n, 1) = 1$ and for all
$t \ge 2, n \ge 1$:
$$T(n, t) = 1 + 2 \sum_{m = 1}^{n-1} T(m, t-1).$$ It can be proved via
induction on $t$ that $T(n, t) \leq n^t$ for all $n, t$. This is
clearly true for $n=1$ or $t = 1$ since $T(n,t)=1$ if either of these
holds true. Otherwise, when $n,t\ge 2$: 
\begin{align*}
  T(n, t) &= 1 + 2 \sum_{m=1}^{n-1} T(m, t-1)\\
          &= 1 + 2 \sum_{m=1}^{n-2} T(m, t-1) + 2T(n-1,t-1)\\
          &= T(n-1, t) + 2T(n - 1, t-1) \leq (n-1)^t + 2(n-1)^{t-1}\\
          &= (n-1)^{t-1} (n + 1) = n^t \left(1 - \frac{1}{n}\right)^{t-1} \left(1 + \frac{1}{n}\right)\\
          &\leq n^t\left(1-\frac{1}{n^2}\right) < n^t.
\end{align*}
Since we spend $n^{O(1)}$ time per node in the recursion tree, the
overall time complexity of our exact algorithm for the original input
with parameter $k$ is $n^{k+O(1)}$.

\section{Clustering under mergeable constraints}
\label{sec:mergeable_clustering}
In this section, we adapt~\Cref{alg:PTAS_MSD_rnd} to design an FPT approximation scheme for \MSD under mergeable constraints. For a mergeable constraint, we say that a clustering $\cC$ is \emph{feasible} if all clusters in $\cC$ satisfy the constraint. We modify the cost function before running \recMSD(\Cref{alg:PTAS_MSD_rnd}) as follows:
$$\cost'(\cC) = \begin{cases}
  \cost(\cC) & \text{if $\cC$ is feasible}\\
  \infty & \text{if $\cC$ is infeasible}
\end{cases}$$

With this modification, we get the following lemma, analogous to~\Cref{lem:msd_prob}. 
\begin{lemma}
\label{lem:approx_msd_mergeable}
    Let $S$ be a non-empty subset of $X$ which satisfies the given mergeable constraint. Let $\opt_i(S)$ denote the cost of optimal feasible $i$-clustering of $S$, and $(\cC_1, \cC_2, \ldots, \cC_t)$ be the output of \recMSD with the modified cost function $\cost'$. Then, for any $\eps \in [0, 1)$ and any $r \in \{1,2, \ldots, t\}$,
    $$\Pr\left[\cost'(\cC_r) \leq \frac{\opt_r(S)}{1 - \eps}\right] \geq \eps^{r-1}.$$
\end{lemma}
In particular, note that the original input $X$ must satisfy the constraint (otherwise no clustering is feasible). The proof of~\Cref{lem:approx_msd_mergeable} follows the proof of~\Cref{lem:msd_prob} with the only additional point being feasibility under the mergeable constraint. If $\opt_r(S) \geq (1 - \eps) \diam(S)$, then the single cluster solution is a feasible $1/(1-\eps)$ approximate solution. Else, with probability at least $\eps$, the random threshold $R$ does not split any cluster of the optimal feasible $r$-clustering. In this event, both subproblems $S_1$ and $S_2$ satisfy the constraint, since they are both unions of optimal clusters. Thus, we can proceed to inductively prove the lemma statement as in the proof of~\Cref{lem:msd_prob}.

The number of nodes in the recursion tree is the same as in~\Cref{alg:PTAS_MSD_rnd}. Since the constraint is verifiable in polynomial time, each node of the recursion tree consumes polynomial time and we get an identical runtime as~\Cref{alg:PTAS_MSD_rnd}. Similarly, applying the same cost function modification to \detRecMSD (\Cref{alg:PTAS_MSD_det}) yields a deterministic FPT approximation scheme for \MSD under mergeable constraints, with an identical running time to~\Cref{alg:PTAS_MSD_det}. This completes the proof of~\Cref{thr:mergeable_clustering}.

\section{Faster FPT approximation scheme for \MSR}
\label{sec:faster_msr}
In this section, we improve the dependence on $k$ in the runtime of our approximation scheme from $(\frac{1}{\eps})^{O(k^2 + \nicefrac{k}{\eps} \log \nicefrac{1}{\eps})}$ to $(\frac{1}{\eps})^{O(\nicefrac{1}{\eps} \log \nicefrac{k}{\eps})}$. 

Consider a subset $S$ of points, an upper bound $t$ on the number of clusters, and the level $\ell$ of the subproblem. Let us call a defensive ball \emph{good} if it has an insignificant core and \emph{bad} otherwise. In case there exists a good defensive ball, in~\Cref{alg:rec_msr}, we included an arbitrary one in our solution and then recursively solved (at the same level $\ell$) for the set of uncovered points. Another place where we made calls at the same level is when solving for the outside subproblem of the expanded ball $E$. Both of these lead to $(\frac{1}{\eps})^{O(k)}$ recursive calls at the same level $\ell$(with a smaller $t$) and overall this leads to an overhead of $(\frac{1}{\eps})^{O(k^2)}$. In order to remove this quadratic term in $k$ in the exponent, we will now directly divide into subproblems at level $\ell + 1$ rather than subproblems at level $\ell$ as well as $\ell + 1$.
\begin{algorithm}[h]
  \caption{{\tt getTemporarySubproblems}}
  \label{alg:expanded_balls_new}
  \DontPrintSemicolon
  \textbf{Global:} The metric space $(X, d)$, the parameters $k, M, L$.\;
  \KwIn{$(\{D_1, D_2, \ldots,D_{t'}\}, (k_1, k_2, \ldots,k_{t'}))$ and the subset $G$ of the indices of good defensive clusters.}
  \KwOut{The set of chosen balls, along with the remaining subproblems.}
  $\cB \gets \{D_j: j \in G\}$ \tcp*{chosen balls, initialized as good defensive balls.}
  $A \gets [t'] \setminus G$ \tcp*{available defensive balls}
  $p \gets 1, P \gets \{\}$ \tcp*{index of the next temporary subproblem and the set of temporary subproblems}
  \While{$A \neq \emptyset$}{
    Let $j$ be an arbitrary index in $A$\;
    $E_p \gets D_j, c \gets c(D_j), r \gets r_j$\tcp*{initialize $E_p$}
    $J_p \gets \{i \in A: D_i \subseteq E_p\}$ \tcp*{available defensive balls fully inside $E_p$}
    $J'_p \gets \{i \in A \setminus J_p: D_i \cap E_p \neq \emptyset\}$ \tcp*{available defensive balls on the boundary of $E_p$}
    \While{$J'_p \neq \emptyset$}{
        $r_{sum} \gets \sum_{i \in J'_p} r_i$\;
        $r_{\max} \gets \max_{i \in J'_p} r_i$\;
        \If{$r_{sum} \ge 8 r_{\max}$ or $r_{\max} \ge \eps^2 \cdot r(E_p)$}{
           $r \gets r +2r_{\max}$, $E_p\gets \Ball(c,r)$ \tcp*{cheap expansion or large expansion}
           $J_p \gets \{i \in A: D_i \subseteq E_p\}$\;
           $J'_p \gets \{i \in A \setminus J_p: D_i \cap E_p \neq \emptyset\}$\;
        }
        \Else{
            \textbf{break}\;
        }
    }
    $\cB \gets \cB \cup \{D_j: j \in J'_p\}$\label{l:add_boundary_balls} \tcp*{add the boundary balls to $\cB$.}
    $k'_p \gets \sum_{j \in J_p} k_j$ \tcp*{number of optimal balls defended by defensive balls with index in $J_p$}
    $S_p \gets ((S \cap E_p) \setminus \bigcup_{B \in \cB} B) \setminus \bigcup_{q=1}^{p-1} E_q \label{l:points_subproblem}$\\
    $P \gets P \cup \{(S_p, k'_p)\}$\tcp*{the $p^\text{th}$ temporary subproblem}
    $p \gets p + 1$\\
    $A \gets A \setminus (J_{p} \cup J'_p)$ \tcp*{Mark the defensive inside or on the boundary of $E_p$ as unavailable}
  }
  \Return $\cB, P$
\end{algorithm}

Given a defensive clustering, our modified algorithm {\tt recMSR2} (with pseudo-code in \Cref{alg:new_rec_msr}) first guesses (by iterating over all $O(2^k)$ choices) the set of defensive balls with insignificant cores and includes all of them in our chosen set of balls $\cB$. Then, it creates a sequence of expanded balls $E_1, E_2, \ldots$ as follows (the pseudocode for this process is given in \Cref{alg:expanded_balls_new}) -- Suppose we have already removed the good defensive balls and created a sequence $E_1, E_2, \ldots,E_{p-1}$. Initially, $p = 1$ and we mark every bad defensive ball as \emph{available}. While there exists an available defensive ball, we initialize $E_p$ as this defensive ball, and then expand it as in~\Cref{alg:expanded_ball}, while only considering available defensive balls. At the end, when no further expansion is possible, we mark each defensive ball inside $E_p$ or on the boundary of $E_p$ as unavailable and increase $p$ by $1$. This gives us a temporary subproblem, consisting of points $S_p$ inside $E_p$ that were not inside any of the $E_i$ for any $i < p$, and were also not inside any of the boundary balls (see line~\ref{l:points_subproblem}). The reason we call this a ``temporary'' subproblem is that we will not include the boundary balls directly in our final solution. Rather we will first expand them as in~\Cref{lem:expansion_for_central_property} in order to satisfy the central property. This results in further pruning of $S_p$.

Consider the set $\cB$ of chosen balls and the temporary subproblems $\{(S_1, k'_1), \ldots, (S_m, k'_m)\}$ corresponding to the expanded balls $E_1, E_2, \ldots, E_m$ respectively. For each $p \in [m]$, let $J_p$ be the set of the defensive balls inside $E_p$ that were available at the time of its creation and $k'_p = \sum_{j \in J_p} k_j$ be the number of optimal balls defended by the defensive balls with index in $J_p$. Let $\eps$ be a small enough constant and recall that $M = \ceil{\frac{1}{\eps^2}}$. We get the following lemma (analogous to~\Cref{lem:constant_balls_cover_constant_fraction}):

\begin{lemma}
\label{lem:constant_frac_new}
     For each $p \in [m]$, there exists a set of $M'_p = \min(M, k'_p)$ optimal balls whose defender balls have index in $J_p$, such that the sum of their radii is at least $\frac{\eps}{6} r(E_p)$.
\end{lemma}
The proof is similar to~\Cref{lem:constant_balls_cover_constant_fraction} after restricting attention to the available balls used to build $E_p$ and observing that all such balls are bad (have a significant core) since the good balls were removed from the available set at the outset.

Similarly, we get the following lemma (analogous to~\Cref{lem:S_in_out_properties}):
\begin{lemma}
    \label{lem:S_p_properties}
    Suppose that the clustering $(\{D_1, D_2, \ldots, D_{t'}\}, (k_1, k_2, \ldots, k_{t'}))$ chosen in line~\ref{l:iterate_over_clusterings} is a defensive clustering of $\OPT(S)$. Then, for every $p \in [m]$, any optimal ball that contains a point from $S_p$ must be defended by a defensive ball with index in $J_p$ and therefore:
    \begin{itemize}
        \item $\OPT(S_{p_1}) \cap \OPT(S_{p_2}) = \emptyset$ for all $p_1 \neq p_2 \in [m]$.
        \item $|\OPT(S_p)| \le k'_p$ for all $p \in [m]$.
        \item $\opt(S_p) \le r(E_p)$ for all $p \in [m]$.
    \end{itemize}
\end{lemma}
\begin{proof}
    Consider an optimal ball $B_i^\star$ that contains a point $x \in S_p$ and let $D_j$ be its defender. Note that $S_p$ is defined as $((S \cap E_p)\setminus \bigcup_{B \in \cB} B) \setminus \bigcup_{i=1}^{p-1}E_i$, where $\cB = \{B_j: j \in G \cup \bigcup_{i=1}^{p} J'_i\}$ is the union of the good defensive balls and the set of boundary defensive balls just after $E_p$ is created (line \ref{l:add_boundary_balls}). First, we claim that $D_j$ must be available when $E_p$ is built. If it was not available, it must have been either a good defensive ball ($j \in G$), or inside/on the boundary of some $E_q$ for $q < p$. In all these cases, $D_j$ would either be in $\cB$ or fully contained within $\bigcup_{q=1}^{p-1} E_q$. Since $x \in B_i^\star \subseteq D_j$, this contradicts $x \in S_p$. Thus, $D_j$ must have been available. Also, $D_j$ can not be fully outside $E$ since it contains a point from $S_p \subseteq E_p$. Also, it can not be on the boundary of $E_p$ since that would mean $D_j \in \cB$, and $S_p \cap B =
    \emptyset$ for all $B \in \cB$ by definition. Thus, $D_j$ must be fully inside $E_p$. By definition, $J_p$ denotes the set of indices of available defensive balls fully inside $E_p$, and hence $j \in J_p$. Once this is known, all three conclusions follow exactly as in~\Cref{lem:S_in_out_properties}.
\end{proof}

\begin{algorithm}[h]
  \caption{{\tt recMSR2} procedure}
  \label{alg:new_rec_msr}
  \DontPrintSemicolon
  \textbf{Global:} The metric space $(X, d)$, the parameters $k, M, L$\;
  \KwIn{A subset $S\subseteq X$, $t \in [k]$, $\ell \in \{0,\ldots,L\}$}
  \KwOut{A clustering of $S$ with at most $t$ balls.}
  \If{$S = \emptyset$}{\Return $\emptyset$\;}
  $\cC \gets {\tt CandidateClusterings}(S, t)$\;

  $\cA^\star \gets $ the cheapest clustering from $\cC$\\
  \If{$\ell = L$}{
    \Return $\cA^\star$;
  }
  \For{\label{l:iterate_over_clusterings}each defensive clustering $(\{D_1, D_2, \ldots,D_{t'}\}, (k_1, k_2, \ldots,k_{t'})) \in \cC$}{
    \For{each $G \subseteq [t']$}{
            $\cB, \{(S_1, k'_1), (S_2, k'_2), \ldots, (S_m, k'_m)\} \gets {\tt getTemporarySubproblems}((\{D_1, \ldots, D_{t'}\}, (k_1, \ldots,k_{t'})), G)$\\
            Let $J \gets \{j_1, j_2, \ldots,j_s\}$ be the set of indices of balls in $\cB$.\\
            \For{tuples $(r'_1, r'_2, \ldots,r'_{s})$ of non negative integers with sum $\leq \lfloor \frac{8k}{\eps}\rfloor$}{
                $D'_q \gets \Ball(c(D_{j_q}), r(D_{j_q}) + r'_q)$ for all $q \in [s]$\\
                $\cB' \gets \{D'_1, D'_2, \ldots,D'_s\}$\label{l:boundary_expanded_balls2}\\
                $\cA \gets \cB'$\\
                \For{$p \in [m]$}{
                    $S_p \gets S_p \setminus \bigcup_{q \in [s]} D'_q$\label{l:pruning_S_p}\\
                    $\cB_p^\star \gets {\tt guessAndSolveMSR2}(S_p, k'_p, \ell + 1)$\\
                    $\cA \gets \cA \cup \cB_p^\star$
                }
                \If{$\cost(\cA) < \cost(\cA^\star)$}{
                    $\cA^\star \gets \cA$
                }
            }
        }
    }
  \Return the best solution from $\cA^\star$;
\end{algorithm}

We now describe how we solve the subproblems returned by~\Cref{alg:expanded_balls_new}. Let $J = G \cup J'_1 \cup \ldots,J'_m$ denote the set of indices of the defensive balls in $\cB$, where $G$ is the set of good defensive balls, and $\bigcup_{p \in [m]} J'_p$ is the set of (indices of) the boundary balls chosen by~\Cref{alg:expanded_balls_new} (in line~\ref{l:add_boundary_balls}). In particular $J'_p$ is the set of indices of the defensive balls on the boundary of $E_p$, that were available at the time of creation of $E_p$. First, we expand the balls in $\cB$ as in~\Cref{lem:expansion_for_central_property} to get a set $\cB'$ (line \ref{l:boundary_expanded_balls2}) of balls that satisfy the central property such that:
\begin{equation}
\label{eqn:Bcost}
    \cost(\cB') \le \sum_{j \in J} r_{\kappa(j)}^\star + \sum_{i \in T} r_i^\star.
\end{equation}
where $T$ denotes the set of (the indices of) optimal balls fully covered by $\cB'$. For each subproblem $p$, we update $S_p$ by removing from it the additional points covered during the conversion from $\cB$ to $\cB'$ (line~\ref{l:pruning_S_p}). Note that, the statements of~\Cref{lem:S_p_properties} hold even after this pruning, since no new optimal balls are added to $\OPT(S_p)$ upon the pruning of $S_p$. We solve each subproblem $p$ independently by making a call to {\tt guessAndSolveMSR2} (\Cref{alg:guess_rec_msr2}) at level $\ell + 1$, in which we first guess the $M'_p = \min(M, k'_p)$ largest optimal balls in $\OPT(S_p)$, and then cover the remaining points using at most $k'_p - M'_p$ clusters via a recursive call to ${\tt recMSR2}$ (at level $\ell + 1$). The total number of balls in the chosen set $\cB$ is exactly $|G| + \sum_{p=1}^m |J'_p|$. Each of these balls is a defensive ball, and thus corresponds to a distinct optimal ball that it defends. For each subproblem $p$, we allocate a budget of $k'_p = \sum_{j \in J_p} k_j$, which is exactly the number of optimal balls defended by the defensive balls with indices in $J_p$. Since the sets $G$, $J'_1, \ldots, J'_m$, and $J_1, \ldots, J_m$ are mutually disjoint subsets of the indices of the defensive clusters (which in turn defend disjoint sets of optimal balls), the total number of balls opened across all subproblems and the set $\cB'$ is at most $\sum_{j=1}^{t'} k_j = t$.

\begin{algorithm}[h]
  \caption{{\tt guessAndSolveMSR2} procedure}
  \label{alg:guess_rec_msr2}
  \DontPrintSemicolon
  \textbf{Global:} The metric space $(X, d)$, the parameters $k, M, L$\;
  \KwIn{A non empty subset $S\subseteq X$, $t \in [k]$, $\ell \in \{1,\ldots,L\}$}
  \KwOut{A clustering of $S$ with at most $t$ balls.}
  \lIf{$|S| = 0$}{
    \Return $\emptyset$
  }
  $\cA^\star \gets null$ \tcp*{assume $\cost(null) = \infty$}
  balls $\gets \{\Ball(p, r): p \in S, r \in \{d(p, q): q \in X\}\}$ \tcp*{set of all possible balls}
  \For{tuples $(B''_1, B''_2, \ldots,B''_{M'})$ of balls with $M' = \min(t, M)$\label{l:guessing_loop2}}{
    $\cA \gets \{B''_1, \ldots,B''_{M'}\} \cup {\tt recMSR2}(S \setminus \bigcup_{j \in [M']} B''_j, t - M', \ell)$\label{l:solution_p}\;
    \If{the cost of $\cA$ is less than the cost of $\cA^\star$}{
      $\cA^\star \gets \cA$\;
    }
  }
  \Return $\cA^\star$\;
\end{algorithm}

\subsection{Approximation factor analysis}
We will prove the following upper bound on the approximation factor:
\begin{lemma}
    \label{lem:faster_msr_approx}
    For each $\ell \in \{0, 1, \ldots,L\}$ and each valid pair $(S, t)$, the cost of the solution returned by ${\tt recMSR2}(S, t, \ell)$ is at most $\rho(\ell)\cdot \opt(S)$ where $\rho(\ell) = 1 + 25 \eps + (1 - \frac{\eps}{6})^{L-\ell}.$
\end{lemma}
\begin{proof}
Our proof will use induction on $|S| + L - \ell$ (analogous to~\Cref{lem:approx_ratio}). The lemma statement trivially holds if $|S| = 0$ since we return an optimal solution in this case, and also if $\ell = L$ since each defensive clustering is a $2$-approximation and $\rho(L) \ge 2$. 

Let us now consider $\ell < L, |S| > 0$. In this case, we obtain a set of defensive balls $\cB$ and a set of temporary subproblems $(S_1, k_1), (S_2, k_2) \ldots, (S_m, k_m)$, corresponding to expanded balls $E_1, E_2, \ldots, E_m$. Here, the set $\cB$ corresponds to indices $J = G \cup J'_1 \cup J'_2 \ldots \cup J'_m$, where $G$ denotes the set of good defensive balls and $J'_p$ denotes the set of boundary defensive balls of $E_p$ (that were available at the time of creation of $E_p$, that is, outside all of $E_1, E_2, \ldots,E_{p-1}$). We then expand the balls in $\cB$ as in~\Cref{lem:expansion_for_central_property} to get a set $\cB'$ of balls that satisfy the central property. Then, we update each $S_p$ by removing the points covered by $\cB'$. Recall from~\Cref{lem:S_p_properties} that $\OPT(S_{p_1})$ and $\OPT(S_{p_2})$ are disjoint for any $p_1 \neq p_2 \in [m]$, and $|\OPT(S_p)| \le k'_p$, $\opt(S_p) \le r(E_p)$ for all $p \in [m]$. Thus, each pair $(S_p, k'_p)$ forms a valid pair since $|\OPT(S_p)| \le k'_p$ and $\cB'$ satisfies the central property. 

To solve the subproblem $S_p$, we consider each tuple $\cB'' = (B''_1, \ldots, B''_{M'_p})$ where $M'_p = \min(k'_p, M)$ and compute ${\tt recMSR2}(S_p \setminus \bigcup_{i \in [M'_p]} B''_i, t - M'_p, \ell + 1)$ (line \ref{l:solution_p}). Let $I'_p$ denote the set of indices of the largest $\min(M'_p, |\OPT(S_p)|)$ optimal balls in $\OPT(S_p)$. Consider the tuple $\cB''$ that consists of all the optimal balls having indices in $I'_p$, along with $M'_p - |I'_p|$ additional balls each having radius $0$ and centered at the center of one of the optimal balls in $I'_p$. Since $(S_p, k'_p)$ is a valid pair, so is $(S_p \setminus \bigcup_{i \in I'_p} B_i^\star, k'_p - M'_p)$ as optimal balls do not contain the centers of other optimal balls and $\OPT(S_p \setminus \bigcup_{i \in I'_p} B_i^\star) \subseteq \OPT(S_p) \setminus \{B_i^\star: i\in I'_p\}$. Thus, by the induction hypothesis the cost of the output $\cB_p^\star$ for the subproblem satisfies:
\begin{equation}
    \label{eqn:B_p_cost}
    \cost(\cB_p^\star) \le \rho(\ell + 1) \opt(S_p \setminus \bigcup_{i \in I'_p} B_i^\star) + \sum_{i \in I'_p} r_i^\star.
\end{equation}

Let $I' = \bigcup_{p \in [m]} I'_p$. From~\Cref{lem:S_p_properties}, note that $I'_p \subseteq \OPT(S_p)$ are all disjoint. Also, let $T$ be the set of indices of optimal balls in $\OPT(S)$ that are fully covered by $\cB'$. Summing up~\Cref{eqn:B_p_cost} and adding $\cost(\cB')$, the final cost of our solution $\cB' \cup \cB_1^\star \cup \ldots \cup \cB_m^\star$ is at most:

\begin{equation}
\label{eqn:begin_analysis}
    \begin{aligned}
    &\cost(\cB') + \sum_{p = 1}^{m} \rho(\ell + 1) \opt(S_p \setminus \bigcup_{i \in I'_p}B_i^\star) + \sum_{i \in I'} r_i^\star\\
    &\leq \sum_{j \in J} r_{\kappa(j)}^\star + \sum_{i \in T} r_i^\star + \sum_{p = 1}^{m} \rho(\ell + 1) \opt(S_p  \setminus \bigcup_{i \in I'_p}B_i^\star) + \sum_{i\in I'}r_i^\star\\
    &\leq \sum_{j \in J} r_{\kappa(j)}^\star + \sum_{i \in T\cup I'} r_i^\star + \rho(\ell)\sum_{p=1}^{m}\opt(S_p \setminus \bigcup_{i \in I'_p}B_i^\star) + (\rho(\ell + 1) - \rho(\ell))\sum_{p=1}^{m} \opt(S_p)\\
    &\leq \sum_{j \in J} r_{\kappa(j)}^\star + \sum_{i \in T\cup I'} r_i^\star + \rho(\ell) \left(\opt(S) - \sum_{i \in T \cup I'} r_i^\star\right) + (\rho(\ell + 1) - \rho(\ell))\sum_{p=1}^{m} \opt(S_p)\\
    &\leq \rho(\ell) \opt(S) + \sum_{j \in J} r_{\kappa(j)}^\star + (\rho(\ell + 1) - \rho(\ell))\sum_{p=1}^{m} r(E_p) - (\rho(\ell) - 1)\sum_{i \in T\cup I'} r_i^\star.
    \end{aligned}
\end{equation}
where the first inequality follows from~\Cref{eqn:Bcost}. The second inequality follows from $\rho(\ell + 1) - \rho(\ell) > 0$ and $\opt(S_p \setminus \bigcup_{i \in I'_p} B_i^\star) \leq \opt(S_p)$. The third inequality follows from the fact that $\OPT(S_p \setminus \bigcup_{i \in I'_p} B_i^\star)$ are all disjoint and the optimal balls with indices in $\bigcup_{i \in T \cup I'}$ do not appear in $\OPT(S_p \setminus \bigcup_{i \in I'_p} B_i^\star)$ for any $p$ (since they were all fully covered and hence all the points within them were removed). The final inequality follows from $\opt(S_p) \le r(E_p)$.

Now, analogous to~\Cref{lem:constant_fraction_saved}, we claim that:
\begin{equation}
\label{eqn:good_coverage_large}
\sum_{i \in T \cup I'} r_i^\star - \sum_{j \in G}\sum_{i \in I(j)} r_i^\star \geq \frac{\eps}{6} \sum_{p=1}^{m} r(E_p).
\end{equation}

Indeed, breaking ties by their index in $\OPT(S)$, consider the $M'_p$ largest optimal balls (out of the $k'_p$ total) optimal balls defended by defensive balls with index in $J_p$. Any of these balls (say $B_i^\star$) that is not covered by $\cB'$ must be one of the $\min(M'_p, |\OPT(S_p)|)$ largest optimal balls in $\OPT(S_p)$ according to the same tie breaking rule. Thus, $i \in T \cup I'$. Also, note that $i \notin \bigcup_{j \in G} I(j) \subseteq T$, since $G$ and $J_p$ are disjoint. Then, applying~\Cref{lem:constant_frac_new} gives~\Cref{eqn:good_coverage_large}.

Also, for the set of boundary balls $J'_p$ of $E_p$ (that were available at the time of creation of $E_p$), we must have(analogous to~\Cref{lem:small_boundary_core}):
\begin{equation}
    \label{eqn:boundaries_small}
    \sum_{j \in J'_p} r_{\kappa(j)}^\star \leq 4\eps^2 r(E_p).
\end{equation}
Summing~\Cref{lem:small_boundary_core} over $p \in [m]$, we get:
\begin{equation}
\label{eqn:sum_cores_small}
    \sum_{j \in J \setminus G} r_{\kappa(j)}^\star = \sum_{p=1}^{m} \sum_{j \in J'_p} r_{\kappa(j)}^\star \leq 4\eps^2 \sum_{p=1}^{m} r(E_p).
\end{equation}
Also, for any good defensive ball $j \in G$, we have $r_{\kappa(j)}^\star \leq \eps \sum_{i \in I(j)}r_i^\star$. Thus, substituting~\Cref{eqn:good_coverage_large,eqn:sum_cores_small} in~\Cref{eqn:begin_analysis}, we get that the cost of our solution is at most:
\begin{equation*}
    \begin{aligned}
        &\rho(\ell) \opt(S) + \sum_{j \in G} r_{\kappa(j)}^\star + (4\eps^2 + \rho(\ell + 1) - \rho(\ell)) \sum_{p=1}^{m} r(E_p) - (\rho(\ell) - 1)\left(\frac{\eps}{6}\sum_{p=1}^{m} r(E_p) + \sum_{j \in G} \sum_{i \in I(j)} r_i^\star\right)\\
        &= \rho(\ell) \opt(S) + \sum_{j \in G} r_{\kappa(j)}^\star - \left(\frac{\eps}{6} (\rho(\ell) - 1) - (\rho(\ell + 1) - \rho(\ell)) - 4 \eps^2\right)\sum_{p=1}^{m}r(E_p) - (\rho(\ell) - 1)\sum_{j \in G}\sum_{i \in I(j)} r_i^\star\\
        &\leq \rho(\ell) \opt(S) + \sum_{j \in G} \left(r_{\kappa(j)}^\star - (\rho(\ell) - 1)\sum_{i \in I(j)} r_i^\star\right)\\
        &\leq \rho(\ell) \opt(S).
    \end{aligned}
\end{equation*}

where the first inequality follows from~\Cref{eqn:rho_eqn} and the second inequality follows from the fact that each defensive ball with index $j \in G$ has an insignificant core, that is $r_{\kappa(j)}^\star \leq \eps \sum_{i \in I(j)} r_i^\star$ and $\rho(\ell) - 1 \geq 25 \eps$ for all $\ell \in [0, L].$
\end{proof}
\subsection{Time complexity of {\tt recMSR2}}
In any call ${\tt recMSR2}(S, t, \ell)$, we consider $(\frac{1}{\eps})^{O(k)}$ candidate clusterings in the set $\cC$. For each of these clusterings, we iterate over $O(2^k)$ choices for the subset of \emph{good} defensive balls, and form a set of up to $t$ temporary subproblems. Then we iterate over all $(1/\eps)^{O(k)}$ possible expansions of the chosen balls, and then make a call for each subproblem to {\tt guessAndSolveMSR2} at level $\ell + 1$. Each invocation of {\tt guessAndSolveMSR2} makes $O(n^{2M})$ calls to {\tt recMSR2} at level $\ell + 1$ corresponding to the guessed optimal balls. Thus, any level $\ell$ call to ${\tt recMSR2}$ makes $(1/\eps)^{O(k)} \cdot 2^k \cdot (1/\eps)^{O(k)} \cdot t \cdot n^{2M} = (1/\eps)^{O(k)} n^{2M}$ calls to {\tt recMSR2} at level $\ell + 1$. Hence, the recursion tree of {\tt recMSR2} has a branching factor of $(1/\eps)^{O(k)} n^{2M}$ and a depth of $L$ as there are $L$ levels. Substituting $L = \ceil{\frac{6}{\eps} \ln \frac{1}{\eps}}$ and $M = \ceil{\frac{1}{\eps^2}}$, we get that the number of nodes in the recursion tree is bounded by
$$ \left(\left(\frac{1}{\eps}\right)^{O(k)} n^{2M}\right)^L = \left(\frac{1}{\eps}\right)^{O(\nicefrac{k}{\eps} \log \nicefrac{1}{\eps})}  n^{O(\nicefrac{1}{\eps^3} \log \nicefrac{1}{\eps})}.$$

Since we take $(1/\eps)^{O(k)} \cdot n^{2M}$ time per node of the recursion tree, the overall running time of {\tt recMSR2} is also $(\frac{1}{\eps})^{O(\nicefrac{k}{\eps} \log \nicefrac{1}{\eps})} \cdot n^{O(\nicefrac{1}{\eps^3} \log \nicefrac{1}{\eps})}$, thus proving~\Cref{thr:mainMSR}.

{\footnotesize \bibliography{references}}

\end{document}